\documentclass[a4paper,twocolumn, accepted=2023-04-11]{quantumarticle}
\pdfoutput=1

\usepackage[dvips]{graphicx}
\usepackage{amsmath,amssymb,amsthm,mathrsfs,amsfonts,dsfont}
\usepackage{multirow}
\usepackage{dcolumn}
\usepackage{bm}
\usepackage{svg}
\usepackage{amsmath}
\usepackage{amssymb}
\usepackage{braket}
\usepackage{physics}
\usepackage{amsthm}
\usepackage{mathtools}
\usepackage{diagbox}
\usepackage[utf8]{inputenc}
\usepackage[english]{babel}
\usepackage{scalerel}[2014/03/10]
\usepackage{stackengine}
\usepackage[noend,noline,ruled,linesnumbered]{algorithm2e}
\usepackage[comma, square, sort&compress, numbers]{natbib}
\usepackage[colorlinks,citecolor=blue,linkcolor=blue,bookmarks=true]{hyperref}

\newcommand{\Ucal}{\mathcal{U}}
\newcommand{\Mcal}{\mathcal{M}}

\newcommand{\comments}[1]{\textcolor{black}{#1}}

\newtheorem{theorem}{Theorem}[]
\newtheorem{conjecture}{Conjecture}[]

\newtheorem{lemma}{Lemma}[]

\newtheorem{definition}{Definition}[]
\newtheorem{task}{Task}[]
\usepackage{braket}

\newcommand{\Acal}{\mathcal{A}}

\newcommand{\Dcal}{\mathcal{D}}
\newcommand{\Kcal}{\mathcal{K}}

\newcommand{\Xcal}{\mathcal{X}}

\newcommand{\Rcal}{\mathcal{R}}

\newcommand{\Rbb}{\mathbb{R}}
\newcommand{\Cbb}{\mathbb{C}}
\newcommand{\Ebb}{\mathbb{E}}

\newcommand{\Hcal}{\mathcal{H}}

\newcommand{\Ord}[1]{\mathcal{O}\left( #1 \right)}

\newcommand{\vabs}[1]{\left\| #1 \right\|}
\newcommand{\pbra}[1]{\left( #1 \right)}
\newcommand{\sbra}[1]{\left[ #1 \right]}
\newcommand{\abra}[1]{\left\langle #1 \right\rangle}
\newcommand{\poly}{\mathrm{poly}}
\newcommand{\NP}{\mathrm{NP}}

\newcommand{\uwa}{Department of Physics, The University of Western Australia, Perth, WA 6009, Australia}
\newcommand{\pku}{Center on Frontiers of Computing Studies, Peking University, Beijing 100871, China}

\begin{document}

\title{Quantum Phase Recognition via Quantum Kernel Methods}


\author{Yusen Wu} 
\affiliation{\uwa}
\author{Bujiao Wu} 
\affiliation{\pku}

\author{Jingbo B. Wang}
\thanks{jingbo.wang@uwa.edu.au}
\affiliation{\uwa}

\author{Xiao Yuan}
\thanks{xiaoyuan@pku.edu.cn}
\affiliation{\pku}

\begin{abstract}
The application of quantum computation to accelerate machine learning algorithms is one of the most promising areas of research in quantum algorithms. In this paper, we explore the power of quantum learning algorithms in solving an important class of Quantum Phase Recognition (QPR) problems, which are crucially important in understanding many-particle quantum systems. We prove that, under widely believed complexity theory assumptions, there exists a wide range of QPR problems that cannot be efficiently solved by classical learning algorithms with classical resources. Whereas using a quantum computer, we prove the efficiency and robustness of quantum kernel methods in solving QPR problems through Linear order parameter Observables. We numerically benchmark our algorithm for a variety of problems, including recognizing symmetry-protected topological phases and symmetry-broken phases. Our results highlight the capability of quantum machine learning in predicting such quantum phase transitions in many-particle systems.

\end{abstract}


\maketitle

\section{Introduction}
The complex nature of multi-particle entanglement has stimulated various powerful classical techniques to study many-particle quantum systems, including density functional theory~\cite{kohn2003electronic},  density matrix renormalization group~\cite{white1992density,white1993density}, quantum Monte Carlo~\cite{becca2017quantum,ceperley1986quantum, foulkes2001quantum,carlson2015quantum} etc. Classical machine learning techniques have recently been considered as an alternative means to study and understand many-particle quantum systems and the associated quantum processes~\cite{carleo2017solving,carrasquilla2017machine,glasser2018neural,torlai2018neural,moreno2020deep,torlai2016learning, schindler2017probing, greplova2020unsupervised, wetzel2017unsupervised}. From numerical perspectives, these works have shown that neural network state ans\"atzes have stronger representation power than conventional tensor networks and may help to solve complex static and dynamical quantum problems. Nevertheless, most of these methods lack of rigorous theoretical guarantees, and whether a learning procedure outperforms original method in solving quantum many-body problems is unknown. To answer this question, Huang et al.~\cite{huang2021provably} explored the power of classical machine learning in classifying some quantum phases of matter and proposed a provable efficient classical method in learning through shadow tomography~\cite{huang2020predicting}. 

However, for quantum many-body systems that possess intricate and long-range entanglement, and if the target quantum phase is determined by a non-local order parameter observable, the required sampling complexity is expected to increase exponentially with respect to the system size. Quantum machine learning has been intensely studied in terms of its expressive ability~\cite{du2022efficient,abbas2021power,holmes2022connecting}, optimization~\cite{mcclean2018barren, wang2021noise, sharma2022trainability, larocca2022diagnosing}, provable quantum advantages~\cite{ huang2022quantum,huang2021information, liu2021rigorous}, as well as potential limitations~\cite{stilck2021limitations}. Meanwhile, recent pioneering experiments on quantum computer processors  ~\cite{boixo2018characterizing,arute2019quantum, zhong2020quantum, huang2022quantum} have demonstrated significant quantum computing advantages in random state sampling~\cite{boixo2018characterizing,arute2019quantum, zhong2020quantum} and density matrix learning problems~\cite{huang2022quantum}. In detail, Huang et al.~\cite{huang2021information, huang2022quantum} proved that the entangled Bell-Measurement protocol can efficiently extract information from an unknown density matrix and predict its linear properties, meanwhile, it is classically hard in the \emph{worst-case} scenario. Therefore, two critical questions are still open: (1) what is the limitation of classical machine learning in solving quantum many-body problems? and (2) whether a near-term quantum computer can enhance the power of classical learning algorithms in solving practical problems?

In this paper, we provide a novel approach to address the above two questions. For a many-body quantum system described by an $n$-qubit parameterized Hamiltonian $H(\bm a)=\sum_{j=1}^m\bm a_jP_j$, where $\bm a\in\mathbb{R}^m$ represents external parameters and $P_j$ are $n$-qubit Pauli operators, we focus on a class of Quantum Phase Recognition problems that can be distinguished by a Linear order parameter Observable, termed as the LO-QPR problem. We aim at learning about detailed phase transitions of many-particle quantum systems using a quantum kernel method.  In the learning phase, a classical training data set $\mathcal{S}=\{(\bm a_{\bm i},b_i)\}_{i=1}^N$ is used, where $\bm{a_i}$ and $b_i$ are respectively the external parameters and ground state property observed from experiments. In the prediction phase, the learning algorithm succeeds if it correctly predicts the property $b$ of the ground state $|\psi(\bm a)\rangle$. For example, considering the Ising Hamiltonian, the external parameter $\bm a$ could be the strength of the transverse magnetic field, and  $b$ represents quantum phases such as the paramagnetic, ferromagnetic, and antiferromagnetic phases. Phase transitions occur when the external parameters varies~\cite{sachdev1999quantum}, and the ability to correctly predict the quantum phase transition boundary can help us understand many strong-correlated systems, even for canonical microscopic physical models~\cite{zheng2017stripe}.

Under two widely accepted assumptions: (1) the polynomial hierarchy does not collapse in the computational complexity theory, and (2) the classical hardness for ground state sampling holds, we prove that certain LO-QPR problems are hard for any classical machine learning methods with classical resources. We demonstrate
that if these LO-QPR problems could be solved by a classical learner with classical resources (even
up to a general additive error tolerance), then the infinite tower of
the polynomial hierarchy would collapse to its second level. While this does not imply that $\rm P=NP$, such a collapse is also widely regarded as being implausible. We therefore answer the first question by showing the exact limitation of classical machine learning in LO-QPR.

The rapid advancement of realistic quantum devices provides an opportunity to answer the second question in a fundamentally different and more powerful
way compared with classical ML. Instead of classically simulating ground states and then infer quantum phase transition, we utilize quantum machine learning to extract high-level abstractions from observed data and directly process quantum ground states information by a quantum computer. Here, the ground state $|\psi(\bm a)\rangle$ of $H(\bm a)$ embeds classical external parameter $\bm a$ onto a specific quantum-enhanced feature space, where inner products of such quantum feature states give rise to a quantum kernel, a metric to characterize distances in the feature space. 
As a result, predicting the ground state property can be transformed into quantum state overlap computation, and thus bypasses the required exponential sample complexity in Ref.~\cite{huang2021provably}. We prove that the proposed Quantum Kernel Alphatron~(QKA) algorithm can efficiently learn from quantum data and solve LO-QPR problems with a promisingly small learning error. We benchmark the proposed QKA in detecting symmetry-protected topological phases and symmetry broken phases, and simulation results show better performances compared with previous QMLs~\cite{cong2019quantum} and classical MLs~\cite{huang2021provably}

This paper is organized as follows. In Sec.~\ref{sec:prelimi} we review some related works and give the definition of the LO-QPR problem, then introduce supervised learning with quantum feature spaces.
We prove the hardness of classical learning algorithms in Sec.~\ref{sec:classicalHardness}.
Sec.~\ref{sec:qalg} details our quantum learning algorithm for LO-QPR problem. Sec.~\ref{sec:numericalRes} presents our numerical simulations. Sec~\ref{sec:complexity} classifies various complexity classes of learning algorithms. Finally, Sec.~\ref{sec:dis} concludes the paper.
\section{Preliminaries}
\label{sec:prelimi}
To clearly demonstrate our contributions in this paper, we first review previous related works and define the learning and computation tasks of interest.

\begin{task}[Density Matrix Learning~\cite{huang2022quantum,huang2021information}] 
\emph{Given an artificial $n$-qubit density matrix $\rho=(I+0.9P)/2^n$ where $P\in\mathcal{P}=\{I,X,Y,Z\}^{\otimes n}\setminus I^{\otimes n}$. The learning algorithms learn about $\rho$ through conventional or quantum-enhanced measurement strategies. The learning algorithm succeeds if it can correctly predict the expectation value ${\rm Tr}\left(\rho Q\right)$ within an $\epsilon$ additive error with $3/4$ probability for $Q\in\mathcal{P}$.}
\label{task:dml}
\end{task}

\noindent In~\cite{huang2022quantum,huang2021information}, the authors proved that the \emph{entangled Bell-Measurement} protocol can use $\mathcal{O}(n/\epsilon^4)$ copies of $\rho$ to solve this learning problem, meanwhile, it is classically hard in the \emph{worst-case} scenario in estimating ${\rm Tr}\left(\rho Q\right)$ for some $Q\in\mathcal{P}$.

\begin{task}[Quantum Phase Learning~\cite{huang2021provably}]
Given the shadow tomography data set $\mathcal{D}=\{\Phi_{\rm shadow}(\bm a_i)\}_{i=1}^N$, where $\Phi_{\rm shadow}(\bm a_i)$ is a classical representation of the ground state of a Hamiltonian $H(\bm a_i)$, the task is to determine the quantum phase of matter for each $\Phi_{\rm shadow}(\bm a_i)\in\mathcal{D}$.
\label{task:qpl}
\end{task}
\noindent 

\noindent In~\cite{huang2021provably}, the authors developed a ``classical'' learning algorithm based on classical data, which are however obtained by shadow tomography of the target ground state generated by a quantum computer. 
As they discussed in~\cite{huang2021provablysup}, when the quantum phase transition can only be determined by a non-local order parameter observable, the sample complexity is expected to increase exponentially with respect to the system size. 
In this case, a better way is using a quantum computer to learn directly from the quantum phase value $b\in\{0,1\}$ observed from the experiment, which will be defined as Task~\ref{LO-QPR} below.  


\begin{figure*}[t]
\centering
  \includegraphics[width=0.8\textwidth]{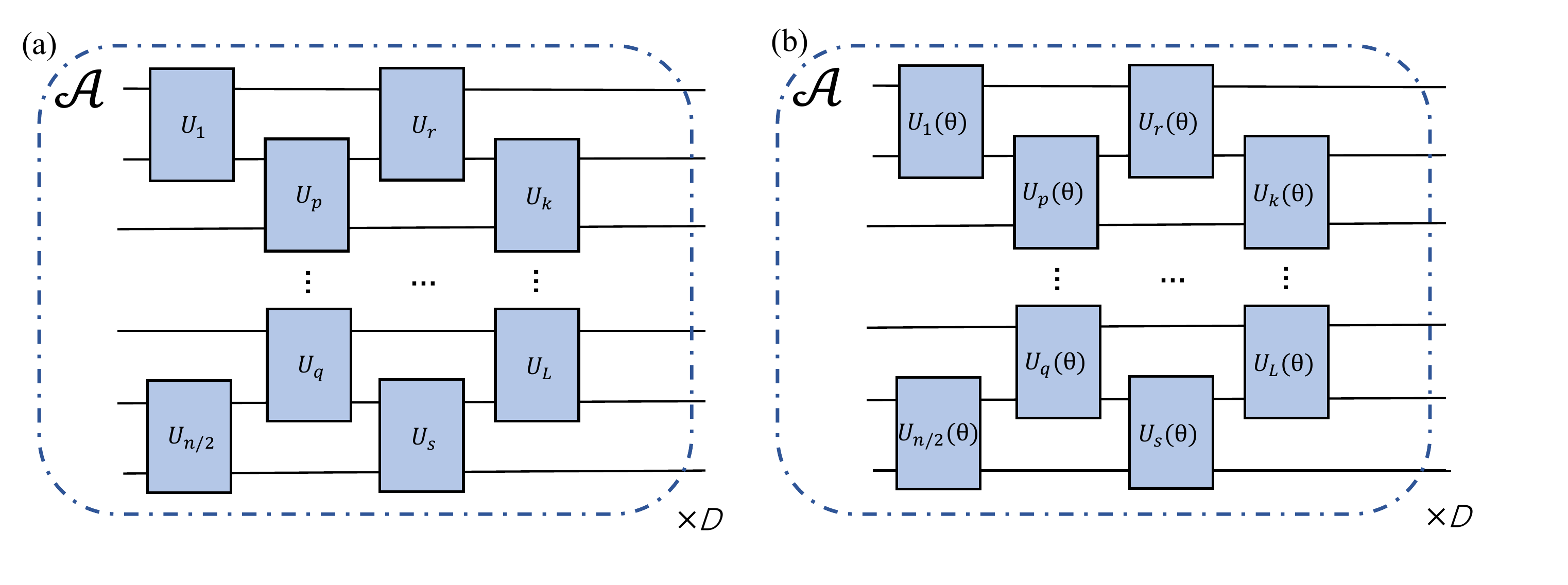}
  \caption{The comparison between quantum random circuits and variational quantum circuits: (a) The quantum random circuit $U=U_1U_2...U_L...U_{LD}$ is generated according to a brickwork architecture $\mathcal{A}$, and each $U_i$ is randomly sampled from $SU(4)$ group based on the Haar measure~\cite{haferkamp2022linear}. (b) The variational quantum circuit $U(\bm\theta)=U_1(\theta_1)U_2(\theta_2)...U_L(\theta_L)...U_{LD}(\theta_{LD})$ can be generated by the same architecture $\mathcal{A}$. Each parameterized quantum gate $U_i(\theta_i)$ is determined by $\theta_i$. Basically, the selection of the $U_i(\theta_i)$ is flexible, including the alternating layered ansatz~\cite{cerezo2021cost}, hardware efficient ansatz~\cite{kandala2017hardware} and particle preserved ansatz~\cite{gard2020efficient}.}
  \label{fig:random_circuit_1}
\end{figure*}

\begin{task}[Ground state Linear Property~(GLP) problem]
    Given an $n$-qubit Hamiltonian $H(\bm a)$ with external parameters $\bm a$ and an observable $\mathcal{M}\in \Cbb^{2^n \times 2^n}$, the goal of GLP is to approximate the ground state property $b=\langle\psi(\bm a)|\mathcal{M}|\psi(\bm a)\rangle$ to additive error $\epsilon=1/{\rm{poly}}(n)$, that is $|\hat{b}-b|\leq\epsilon$, where $\hat{b}$ represents the estimated property and $|\psi(\bm a)\rangle$ is the ground state of $H(\bm a)$.
\end{task}
Noting that the GLP problem could characterize a class of quantum phase problem that is determined by a linear order parameter, which is one of the most significant elements in understanding quantum fluctuation phenomenons in condensed-matter systems~\cite{ haldane1983nonlinear, pollmann2012detection}. For example, considering the Ising Hamiltonian, the external parameter $\bm a$ and the order parameter observable $\mathcal{M}$ could be the strength of the transverse magnetic field and the spin correlation respectively, and quantum phases $b$ include paramagnetic, ferromagnetic, and antiferromagnetic phases. In general, it would be hard to recognize quantum phases of an arbitrary many-body quantum system, owing to the hardness of obtaining the ground state and the fact that the order parameter is generally unknown. Nevertheless, there may also exist cases where the problem is exactly efficiently solvable for very specific choices of parameters. Then, a natural question is, based on the solvable or known quantum phases from experiments, whether we could learn and predict GLP for other external parameter domains. Therefore, it is reasonable to consider the learning version of the GLP problem.

\begin{task}[LO-QPR]
\label{LO-QPR}
Given the training data $\mathcal{S}=\{(\bm a_{\bm i},b_i)\}_{i=1}^N$ for which $\bm{a_i}$ indicates the external parameter, and $b_i$ represents its quantum phase value characterized by some unknown observable, associated with an $n$-qubit Hamiltonian $H\pbra{\bm{a_i}}=\sum_k\bm a_i^{(k)}P_k$, our aim is to efficiently learn a model $h(\bm a)$ enabling the risk
\begin{align}
R(h)=\sum\limits_{\bm a\sim\mathcal{X}}\mathcal{D}(\bm a)\left(h(\bm a)-b\right)^2
\end{align}
is upper bounded by $\mathcal{O}\left(\sqrt[4]{\frac{\log (1/\delta)}{N}}\right)$ with $1-\delta$ success probability for some fixed distribution $\mathcal{D}(\bm a)$ defined on the external parameter space $\mathcal{X}$.
\label{def:QPRLearn}
\end{task}

The crux of the matter on solving LO-QPR problems can be summarized into two levels: how to generate high-quality training data and how to efficiently learn from these data.
Solving LO-QPR problems inevitably involves high dimensional Hilbert space, and classical learning algorithms thus are challenging in both generating and operating training data with quantum entanglement.


\section{Classical hardness for LO-QPR}
\label{sec:classicalHardness}
\subsection{Express ground state by quantum circuit} 
\label{sec:meas_distribute}

{The \emph{brickwork} architecture $\mathcal{A}$ is a general approach to construct ansatz for ground state computation~\cite{haferkamp2022linear}, whose structure is formed as follows: perform a string of two-qubit gates $U_1\otimes U_2\otimes...\otimes U_{n/2}$ as the first layer, then perform a staggering string of gates, as illustrated as Fig.~\ref{fig:random_circuit_1}~(a). Meanwhile, the brickwork architecture can induce structured variational quantum circuit, which is a convinced method in approximating the ground state of many-body Hamiltonians $H(\bm x)$~\cite{peruzzo2014variational,kandala2017hardware, cade2020strategies, mcardle2019variational, li2022toward,cao2021larger}.

The key idea of using variational quantum circuit is that the parameterized quantum state is prepared and measured on a quantum computer, and the classical optimizer updates the parameter $\bm \theta$ according to the measurement information. With the brickwork architecture $\mathcal{A}$, the variational ansatz can be prepared by $|\Psi(\bm \theta)\rangle=U(\bm\theta)|0^{n}\rangle=\prod_{d=1}^DU_d(\bm\theta_d)|0^{n}\rangle$, where $U(\bm\theta)$ is composed of $D$ unitaries $U_d(\bm\theta_d)$ whose structure is shown in Fig.~\ref{fig:random_circuit_1}~(b). After several classical optimisation steps, the classical optimizer can provide a parameter $\bm\theta_{\bm x}$ enabling $|\Psi(\bm\theta_{\bm x})\rangle$ be an approximation of $|\psi(\bm x)\rangle$. We provide a deterministic method in finding a $\bm\theta_{\bm x}$ to approximate $|\psi(\bm x)\rangle$ by using the quantum imaginary time evolution in the Appendix~\ref{app:archHaar}. Since the variational quantum circuit $U(\bm\theta)$ has the same architecture $\mathcal{A}$ to that of random circuit $U$, and two-qubit gates $U_i(\theta_i)$ are sampled from a subset of ${\rm SU}(4)$. Then the relationship $\mathcal{U_A}(\bm\theta)\subseteq\mathcal{U_A}$ holds, where $\mathcal{U_A}(\bm\theta)$ and $\mathcal{U_A}$ denote the set of $U(\bm\theta)$ and random circuit set, respectively.}

\subsection{Classical hardness results}
The GLP problem is an instance of the mean-value problem which is the central part of the variational quantum algorithms, and ``\emph{Is the quantum computer necessary for the mean value problem?}'' is still an open problem as mentioned in~\cite{bravyi2021classical}. Note that for
a quantum state $|\psi(\bm x)\rangle=U|0^n\rangle$, Ref~\cite{bravyi2021classical} proposed an upper bound on estimating $\langle \psi(\bm x) |\mathcal{M}|\psi(\bm x) \rangle$ in the case of a ${\rm{poly}}(n)$-depth $U$ associated with $\ket{\psi(\bm x)}$. Here, we show the reduction from quantum circuit sampling problem to the GLP problem (see the proof in Lemma~\ref{lem:LowerCompute}), and prove that the classical hardness of quantum sampling will imply the hardness of the GLP problem in the worst-case scenario, and thus provide a lower bound on estimating $\langle \psi(\bm x) |\mathcal{M}|\psi(\bm x) \rangle$. 

The proof of Lemma~\ref{lem:LowerCompute} is inspired by the classical hardness conjecture of the random quantum state sampling problem in~\cite{bouland2019complexity}. The random quantum state sampling is an artificial problem that samples from the output distribution of some experimentally feasible quantum algorithms. Sample one particular bit-string $\bm j$ from a random quantum circuit $U$ with the exact probability $p_{U}(\bm j)$ is believed classically hard under standard complexity assumptions, which is also recognized as the worst-case hardness.

However, a convincing quantum advantage must be established in the average-case scenario, that is, the classical hardness should be held across the entire distribution, rather than concentrated in a single quantum process and output. Bouland et al.~\cite{bouland2019complexity} introduced the Feynman path integral to connect a bridge between a fixed outcome $\bm j$ and a low-degree multivariate polynomial described by quantum gates, and thus proved the worst-to-average reduction. We utilize a similar Feynman path integral method to prove the classical learning hardness result for LO-QPR problem (see Theorem \ref{thm:LowerLearn}).

\begin{conjecture}[Ref.~\cite{bouland2019complexity}]
There exists an $n$-qubit ground state $|\psi(\bm x)\rangle=U(\bm\theta_{\bm x})|0^n\rangle$ such that the following task is \#P-hard: approximate $p_U(\bm j)=\abs{\langle\bm j |\psi(\bm x)\rangle}^2$ to additive error $\epsilon_{c}/2^n$ with probability $\frac{3}{4}+\frac{1}{\poly(n)}$, where $\bm j$ is a $\{0,1\}^n$ bit string, $U(\bm\theta_{\bm x})$ is an $n$-qubit quantum circuit and $\epsilon_c=1/{\rm{poly}}(n)$ \footnote{Since any quantum state can be encoded as a ground state of some Hamiltonian, the presented conjecture uses the `ground state' to substitute `random quantum state' mentioned in Ref.~\cite{bouland2019complexity}}.
\label{con:lowerApprox}
\end{conjecture}

\noindent \comments{The probability distribution of a truly random quantum state $|\psi\rangle$ possesses the Porter-Thomas (PT) distribution ${\rm Pr}(|\langle\bm j|\psi\rangle|^2)=2^ne^{-2^n |\langle\bm j|\psi\rangle|^2}$, which is known to be classically hard to sample~\cite{brody1981random, boixo2018characterizing}. Whether ground states $|\psi(\bm x)\rangle$ of a family of Hamiltonian $H(\bm x)$ satisfy conjecture~\ref{con:lowerApprox} can be verified by comparing probability distribution of $|\psi(\bm x)\rangle$ to the PT distribution. We show that if the probability distribution of a ground state $|\psi(\bm x)\rangle$ is $\mathcal{O}(n^{-1})$-close to PT distribution by means of trace distance, then sample from  $|\psi(\bm x)\rangle$ is classically hard. Details refer to Theorem~\ref{thm:verify} in the Appendix~\ref{App:SamplingComplexity}.}
Using this result and the relationship between ground states and variational quantum states, the hardness result for the GLP problem can be stated as the following lemma.

\begin{lemma}
With the assumption that Conjecture~\ref{con:lowerApprox} holds, and the Polynomial Hierarchy (PH) in the computational complexity theory does not collapse, there exists an $n$-qubit Hamiltonian $H(\bm a)$ and an observable $\mathcal{M}$, such that their corresponding GLP problem cannot be efficiently calculated by any classical algorithm.
\label{lem:LowerCompute}
\end{lemma}

\noindent We provide the detailed proof in the Appendix~\ref{proof_lemma_lower_compute}. This lemma also serves for the hardness of the LO-QPR problem. Let $|\psi(\bm a)\rangle=U(\bm\theta_{\bm a})|0^{n}\rangle$ be the ground state of Hamiltonian $H(\bm a)$ satisfying Lemma \ref{lem:LowerCompute}, where $U(\bm\theta_{\bm a})\in\mathcal{U_A}(\bm\theta)$. For a general family of Hamiltonian $H(\bm a)$, the explicitly mathematical expression of its ground state is hard to be characterized, and variational quantum circuit is a natural expression which characterizes the circuit complexity of these concerned ground states, but not change their properties~\cite{kandala2017hardware, cade2020strategies, peruzzo2014variational,mcardle2019variational}.
As a starting point of worst case scenario $|\psi(\bm a)\rangle=U(\bm\theta_{\bm a})|0^{n}\rangle$, we provide a method on constructing ground states $\mathcal{T}=\{|\psi(\bm a_i)\rangle\}$ by using variational circuit $U(\bm\theta)$. Detail refers to the Appendix~\ref{proof_them_1}. Then we have the following theorem for the average case hardness of the LO-QPR problem on $\mathcal{T}$. Given classical training data $\mathcal{S}$ (given $\bm a_i$, its label $b_i$ can be computed by some classical Turing machines), we prove that no classical learning algorithm can efficiently learn the hypothesis $h^*$ such that $R(h^*(\bm x))$ is upper bounded by $1/{\rm poly}(n)$ for $\bm x\in\mathcal{T}$.
\begin{theorem}
Given training data $\mathcal{S}=\{(\bm a_{\bm i},b_i)\}_{i=1}^N$ (acquired from classical methods) 
for which $(\bm{a_i}, b_i)$ indicate the external parameter and phase value associated with the $n$-qubit Hamiltonian $H(\bm{a_i})$, 
there exists a testing set
$\mathcal{T}=\{(\bm x_{\bm i}, y_i)\}_{i=1}^M$, such that predicting $8/9$ of $y_i\in\mathcal{T}$ with additive error $1/{\rm poly}(n)$ is hard for any classical ML algorithm,
with the assumption that Conjecture~\ref{con:lowerApprox} holds and the PH does not collapse, where the scale of testing data $M={\rm poly}(N)$ and $N={\rm poly}(n)$.
\label{thm:LowerLearn}
\end{theorem}

Theorem~\ref{thm:LowerLearn} gives average-case hardness results for the C-Learning Alg.+ C-Data on the testing set $\mathcal{T}$. The meaning of C-Learning Alg.+C-Data refers to Sec.~\ref{sec:complexity}, and proof details are provided in the Appendix~\ref{subapp:LowerLearn}. 

\begin{figure*}[t]
\centering
  \includegraphics[width=0.9\textwidth]{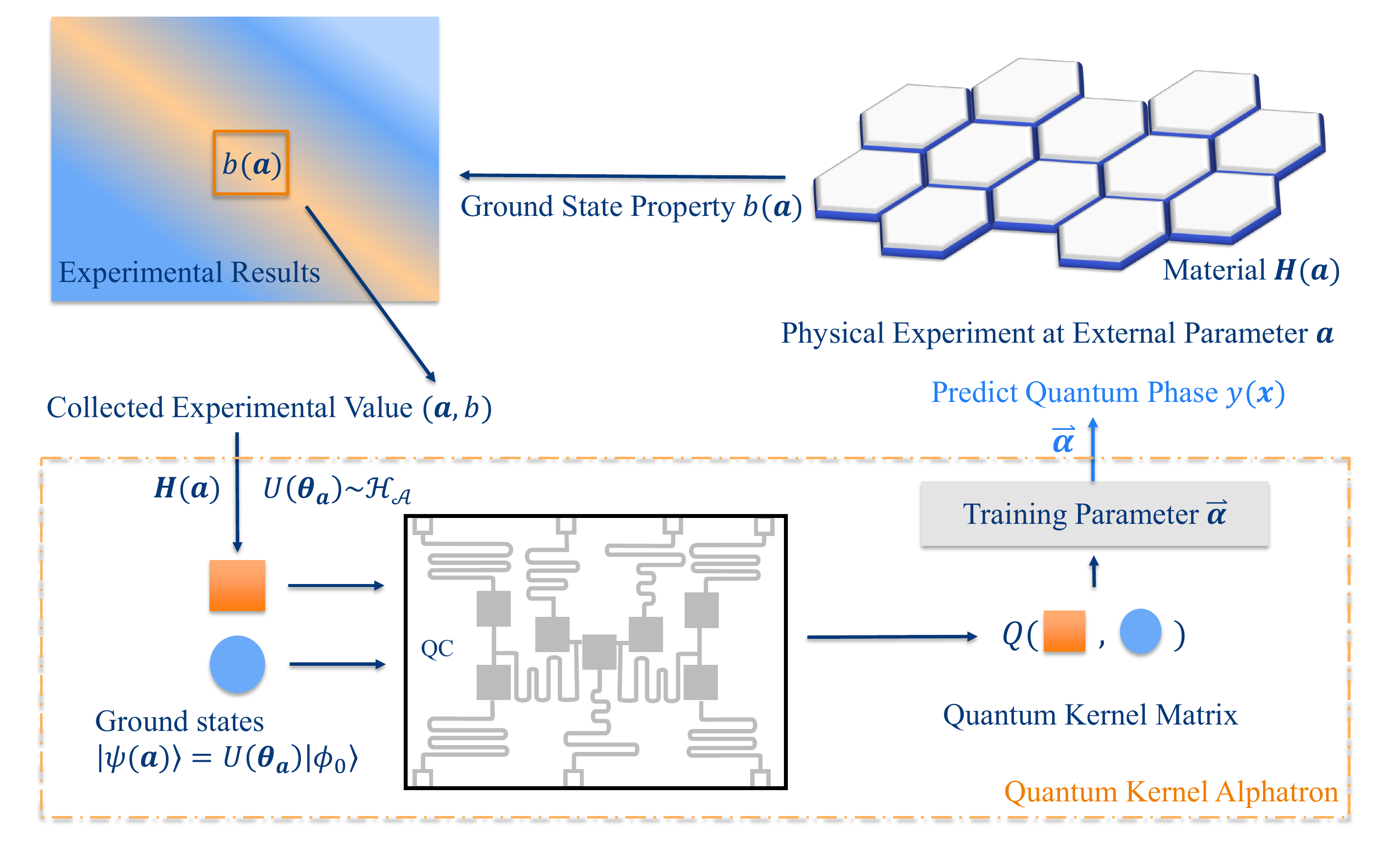}
  \caption{
  The procedure of the proposed quantum learning algorithm. Here, $|\phi_0\rangle$ denotes the initial state of the quantum system.}
  \label{fig:kernel}
\end{figure*}
\vspace{8pt}

\section{Quantum learning algorithm}
\label{sec:qalg}

\subsection{Supervised learning with 
quantum feature space}
Here, we denote $(\bm a,b)$ (or $(\bm x, y)$) as a pair of the datum $\bm a$ ($\bm x$) and the corresponding label $b$ ($y$) in the training set $\mathcal{S}$ (testing set $\mathcal{T}$).
Generally, the task of supervised learning is to learn a label $y$ of the testing datum $\bm x\in\mathcal{T}\subset\mathcal{X}$ from a distribution $\mathcal{D}(\bm x)$ defined on the space $\mathcal{X}$ according to some decision rule $h$. The decision rule $h$ is assigned by a selected machine learning model from the training set $\mathcal{S}=\{(\bm a_{\bm i},b_i)\}_{i=1}^N$, where $\bm a_{\bm i}\in \Xcal$ 
follows distribution $\mathcal{D}(\bm a_{\bm i})$, the label $b_i=h(\bm a_{\bm i})$, and $N$ is the size of the training set.
Given the training set $\mathcal{S}$, an efficient learner needs to generate a classifier $h$ in 
poly$(N)$ time, to minimize the error
\begin{align}
R(h) = \Pr_{\bm x\sim \Dcal}[h(\bm x)\ne y].
\label{Eq:generalized_error}
\end{align}
The datum $\bm x$ is sampled randomly according to $\mathcal{D}(\bm x)$, in both the training and testing procedure, and the size $N$ of the training set is polynomial in the data dimension.

The kernel method has played a crucial role in the development of supervised learning~\cite{mohri2018foundations,goel2019learning,kalai2009isotron}, which provides an approach to increase the expressivity and trainability of the original training set.
We can describe a kernel function $\Kcal:\mathcal{X}\times \mathcal{X} \rightarrow \Rbb$ as $\Kcal(\bm x, \bm x')=\widetilde{\Psi}(\bm x)^{T}\widetilde{\Psi}(\bm x')$, where $\widetilde{\Psi}: \Xcal \rightarrow \Hcal$ is the feature map
which maps the datum $\bm x\in\mathcal{X}$ to a higher-dimensional space $\Hcal$ (feature space).
Tremendous classical kernel methods~\cite{kalai2009isotron, goel2019learning} have been proposed to learn the non-linear functions or decision boundaries. With the rapid development of quantum computers, there is a growing interest in exploring whether the quantum kernel method can surpass the classical kernel~\cite{schuld2019quantum,blank2020quantum,schuld2021quantum, shaydulin2021importance,liu2021representation,shirai2021quantum, nakaji2021quantum, jerbi2021quantum, sancho2021quantum, li2021quantum, henry2021quantum, kubler2021inductive, glick2021covariant,hubregtsen2021training,klus2021symmetric,wang2021towards,huusari2021entangled,kusumoto2019experimental}.
Here we leverage the \emph{quantum kernel} as our kernel function--$Q(\bm x, \bm x')=\abs{\langle\psi(\bm x)|\psi(\bm x')\rangle}^2$, where $\ket{\psi(\bm x)}$ is the ground state of $H(\bm x)$. 

\subsection{Quantum Kernel Alphatron}
Here, we show the possibility of solving the LO-QPR problem with quantum data by leveraging the quantum kernel method combined with the Alphatron algorithm~\cite{goel2019learning}. 
From the learning theory perspective, training can be phrased as the empirical risk minimization, and the associated learning model $h^\ast$ for the minimise of the empirical risk follows the representer theorem.

\begin{algorithm}[htbp]
\SetKwInOut{Input}{Input}
\SetKwInOut{Output}{Output}
\SetKwFor{While}{while}{do}{}%
\SetKwFor{For}{for}{do}{}
\SetKwIF{If}{ElseIf}{Else}{if}{do}{elif}{else do}{}%
\Input{training set $\mathcal{S}=\{(\bm a_{\bm i},b_i)\}_{i=1}^N\in\mathbb{R}^d\times [0,1]$, 
variational quantum circuit $U\pbra{\bm\theta}\in\mathcal{U_A}(\bm\theta)$, quantum kernel function $\hat{Q}(\bm a_i.\bm x)=|\langle\psi(\bm a_i)|\psi(\bm x)\rangle|^2$, learning rate $\lambda>0$, number of iterations $T$, testing data $\mathcal{T}=\{(\bm x_{\bm j},y_j)\}_{j=1}^M\in\mathbb{R}^d\times [0,1]$}
\Output{$\hat{h}^r$}

\For{$i=1,2,...,N$}{
\emph{Prepare quantum state $|\psi(\bm{a_i})\rangle=U(\bm\theta_{\bm a_i})|0^{\otimes n}\rangle$, where $\bm\theta_{\bm a_i} := \arg\min_{\bm\theta} \langle 0^{\otimes n}|U^{\dagger}(\bm\theta)|H(\bm{a_i})|U(\bm\theta)|0^{\otimes n}\rangle$}\;
}
$\bm \alpha^{\bm 1}:=0\in\mathbb{R}^N$\;
\For{$t=1,2,...,T$}{
$\hat{h}^{t}(\bm x):=\sum_{i=1}^N\alpha_i^t \hat{Q}(\bm a_{\bm i},\bm x)$\;
\For{$i=1,2,...,N$ }{
 $\alpha_i^{t+1}=\alpha_i^t+\frac{\lambda}{N}(b_i-\hat{h}^t(\bm a_{\bm i}))$\;
}
}
\emph{Let $r=\arg \min_{t\in\{1,...,T\}}\sum_{j=1}^M(\hat{h}^t(\bm x_j)-y_j)^2$}\;
\Return{$\hat{h}^r$}
\caption{Quantum Kernel Alphatron (QKA)}\label{alg:QKernel_alphatron}
\end{algorithm}

\begin{theorem}[Representer theorem~\cite{mohri2018foundations}]
Let $\mathcal{S}=\{(\bm a_{\bm i},b_i)\}_{i=1}^N$ and corresponding feature states $\{|\psi(\bm a_i)\rangle\}_{i=1}^N$ be the training data set.
$Q:\mathcal{X}\times\mathcal{X}\mapsto\mathbb{R}$ be a quantum kernel with the kernel space $\Hcal$. Consider a strictly monotonic increasing regularisation function $g:[0,\infty)\mapsto\mathbb{R}$, and regularised empirical risk
\begin{align}
    \hat{R}_L(h^\ast)=\frac{1}{N}\sum\limits_{i=1}^N\pbra{h^\ast(\bm a_{\bm i})-b_i}^2+g\left(\|h^\ast\|_\mathcal{H}\right).
\end{align}
Then any minimizer of the empirical risk $\hat{R}_L(h^\ast)$
admits a representation of the form $h^\ast(\bm x)=\sum_{i=1}^N\alpha_iQ(\bm a_{\bm i},\bm x)$, where $\alpha_i\in\mathbb{R}$ for all $i\in\{1,2,...,N\}$, $\bm x$ and $\bm a$ are drawn from the same distribution. 
\label{thm:representer}
\end{theorem}

\begin{figure*}[htb]
\centering
  \includegraphics[width=0.85\textwidth]{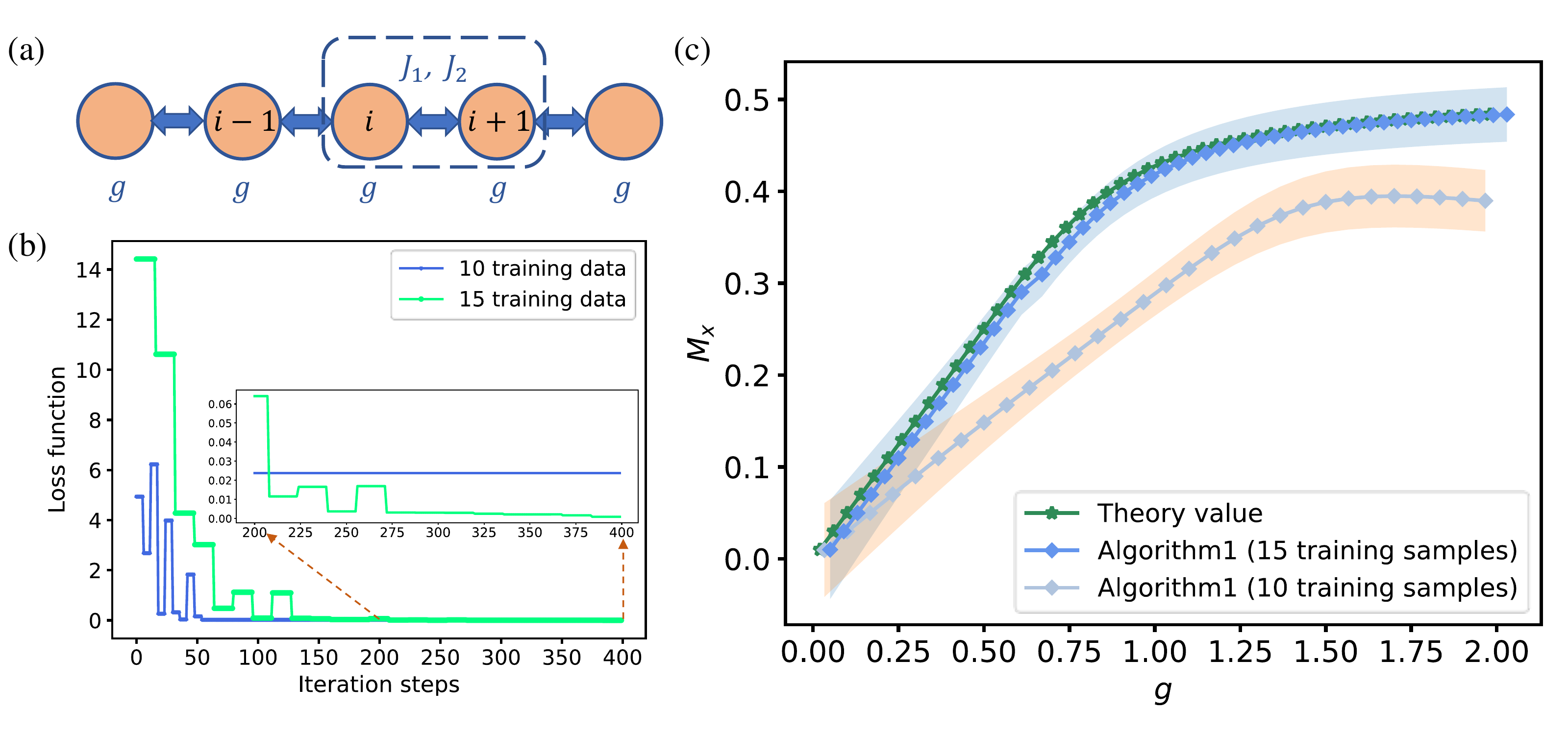}
  \caption{Numerical results for predicting ground-state properties in a $S=1/2$ XXZ spin model with $n=16$ qubits. (a) Illustration of the concerned spin model geometry. (b) The two curves depict the tendency of $R_L(h)$ in the training procedure, where the green (blue) curve represents the number of training data $N=15$ ($N=10$). (c) For fixed iterations (400), we randomly select the training data ($N=10$ or $N=15$) over $10$ trials and plot the average performance by Alg.~\ref{alg:QKernel_alphatron} in predicting the magnetization $M_x$.}
  \label{Fig:QPR_demo_1}
\end{figure*}

\begin{figure*}[htb]
\centering
  \includegraphics[width=1.0\textwidth]{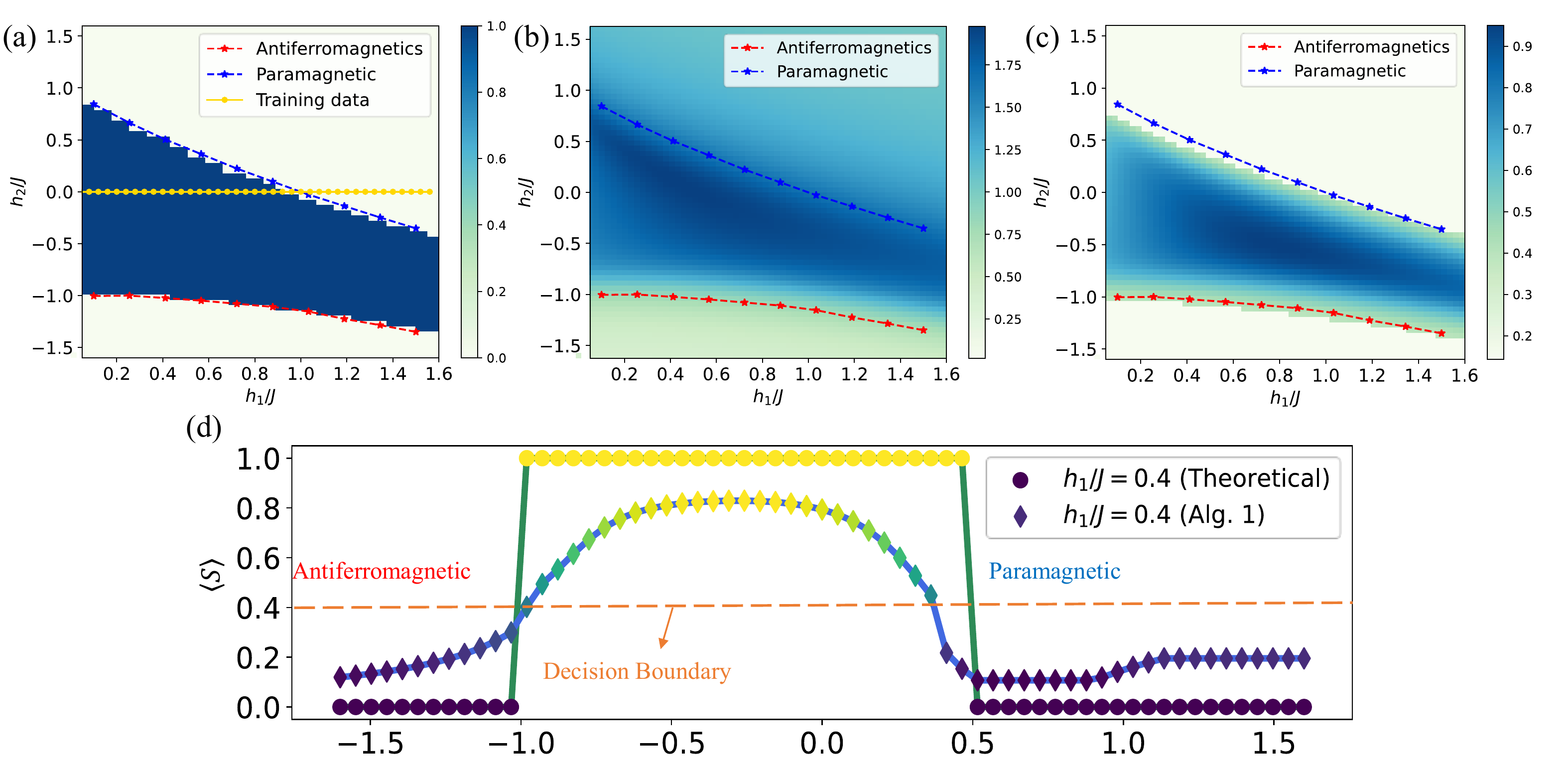}
  \caption{Numerical results for recognizing a  $Z_2\times Z_2$ Symmetry-Protected-Topological (SPT) phase of the Haldane Chain. (a) The exact phase diagram of the Haldane Chain (Eq.~\ref{Eq:HaldaneChain}), where the phase boundary points (blue and red curves) are extracted from the Ref~\cite{cong2019quantum}, and the background shading represents the phase function of $(h_1/J,h_2/J)$. The yellow line indicates $40$ training points on the line $h_2=0$. \comments{(b) The classification result by using $N=15$ training data. (c) The quantum phase classification result of Haldane Chain by using Alg.~\ref{alg:QKernel_alphatron} with $N=40$ training data. Here, the classification accuracy $v_s=0.985$ which has a better performance compared to the $2$-layer QCNN method~\cite{cong2019quantum} ($v_s=0.971$, detail refers to SI material). (d) The quantum phase function at cross-section $h_1/J=0.4$ of the Haldane Chain proposed by Alg.~\ref{alg:QKernel_alphatron}.}
  }
  \label{Fig:QPR_demo_2}
\end{figure*}
\begin{figure*}[htb]
\centering
  \includegraphics[width=0.85\textwidth]{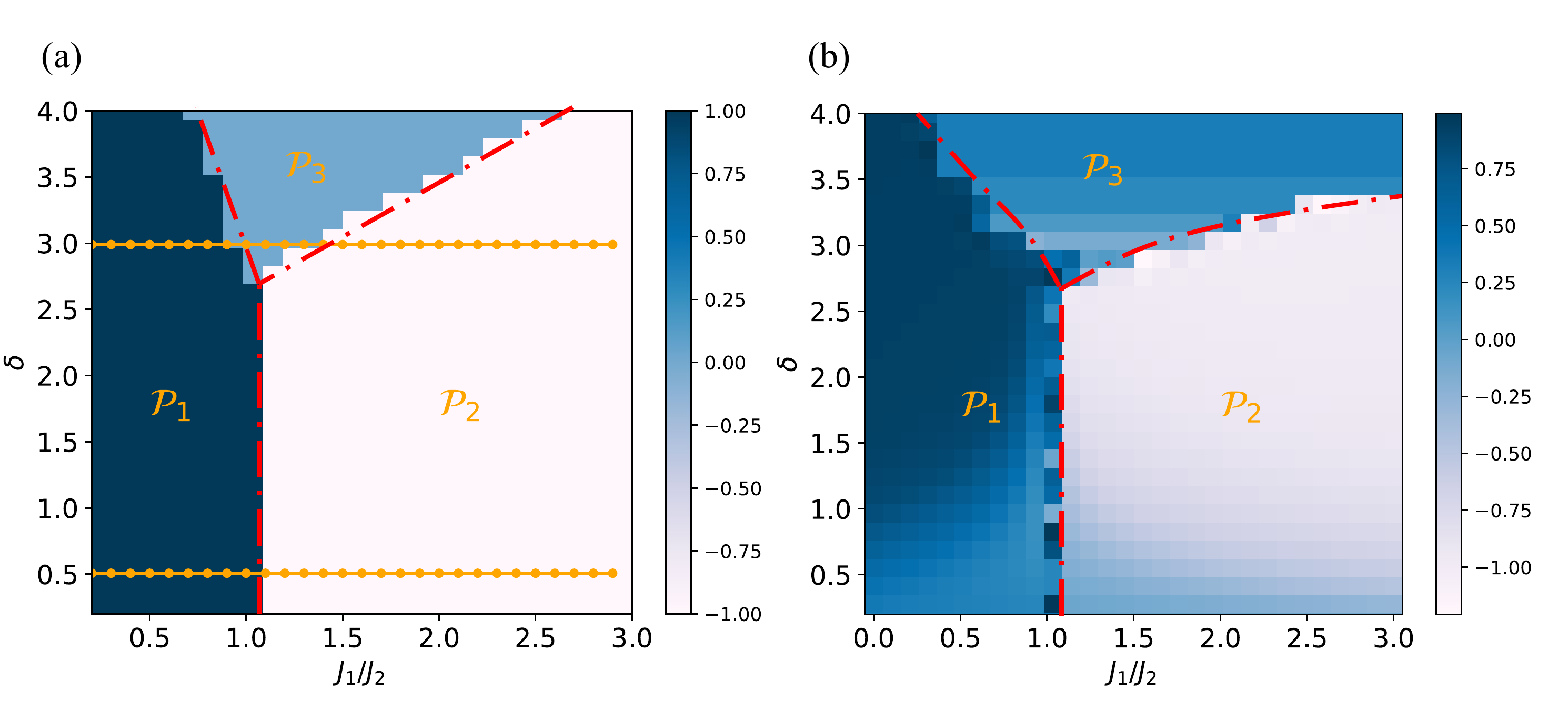}
  \DeclareGraphicsExtensions.
   \caption{Numerical results for recognizing three distinct phases of bond-alternating XXZ model. (a) The system's three distinct phases are characterized by the topological invariant $Z_R$ discussed in the Ref~\cite{elben2020many}. The invariant $Z_R=+1$ marks the Trivial phase  $\mathcal{P}_1$, $Z_R=-1$ marks the  Topological phase $\mathcal{P}_2$ and $Z_R=0$ marks the Symmetry broken phase $\mathcal{P}_3$. Here, the (white and blue) background shading indicates the theoretical quantum phase value, and the red curves depict the quantum phase transition boundary. The two orange lines $\delta=3.0$ and $\delta=0.5$ represent $60$ training points. (b) The quantum phase diagram predicted by the Alg.~\ref{alg:QKernel_alphatron}. }
  \label{Fig:QPR_demo_3}
\end{figure*}
\begin{figure*}[htb]
\centering
  \includegraphics[width=1.0\textwidth]{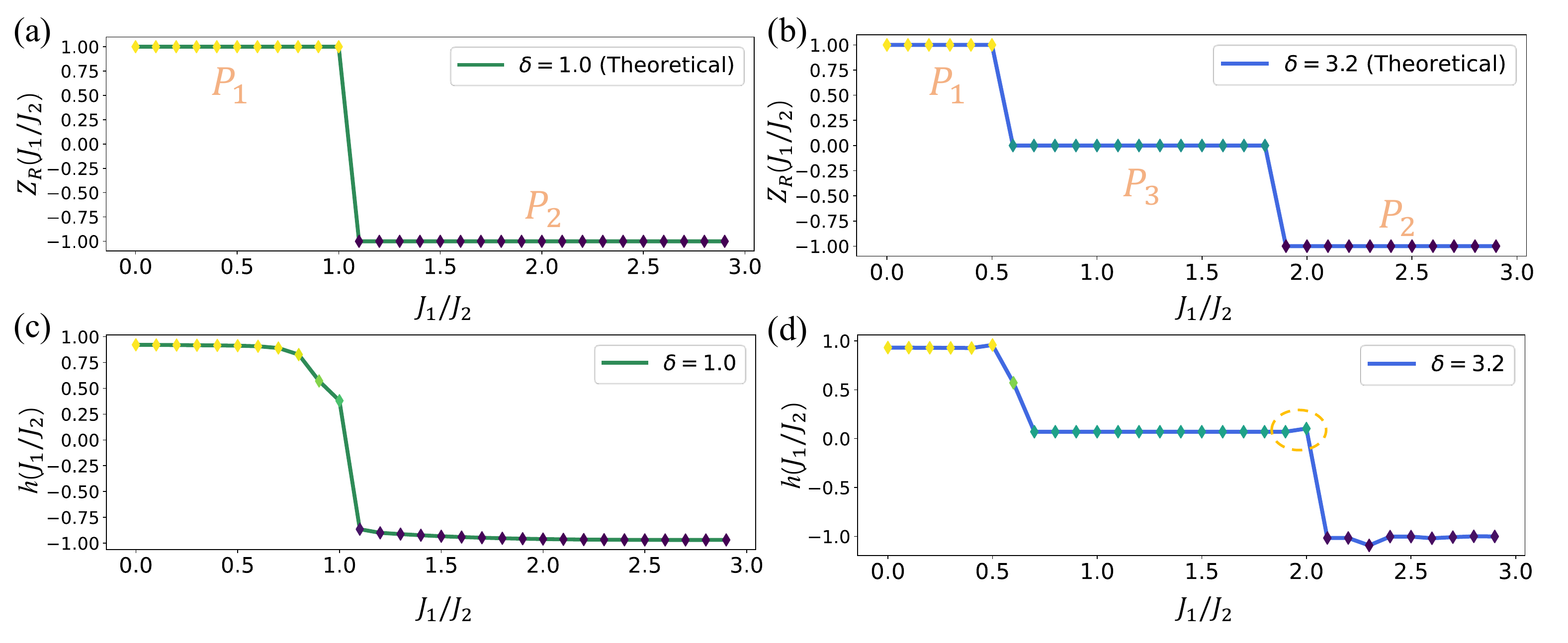}
  \DeclareGraphicsExtensions.
  \caption{The function $Z_R(J_1/J_2)$ at cross sections $\delta=3.2$ and $\delta=1.0$ of the bond-alternating XXZ model phase diagram. Figures (a) and (b) are extracted from Fig.~\ref{Fig:QPR_demo_3} (a) represent the theoretical value on the two lines, and Figures (c) and (d) are extracted from Fig.~\ref{Fig:QPR_demo_3} (b) represent the predictive results. The dotted circle marks the mis-classification ground states.}
  \label{Fig:QPR_demo_3_1}
\end{figure*}

According to the above theorem, one of the options to the quantum kernel is $Q(\bm a_{\bm i},\bm x)=\abs{\langle\psi(\bm a_{\bm i})|\psi(\bm x)\rangle}^2$ where $|\psi(\bm a_i)\rangle$ represents the ground state of $H(\bm a_i)$, and unknown order parameter observable $\mathcal{M}$ thus can be represented as a linear combination of feature states, that is, $\mathcal{M}\approx\sum_i\alpha_i|\psi(\bm a_{\bm i})\rangle\langle\psi(\bm a_{\bm i})|$. Given the kernel matrix $\mathcal{Q}=[Q(\bm{a_i}, \bm a_j)]_{N\times N}$, the optimal weight parameters $\alpha_i$ in the expression of $\Mcal$ has a closed-form solution by leveraging  linear regression algorithms, and it requires $\mathcal{O}(N^{2.373})$ time complexity for solving such problem~\cite{le2014powers}.

Here, we provide a more advanced method in learning $\alpha_i$. From the learning theory perspective, training can be phrased as the empirical risk minimization, and the associated learning model $\hat{h}^\ast$ for the minimizer of the empirical risk can be learned as Alg.~\ref{alg:QKernel_alphatron}. As a result, given the quantum kernel matrix $\mathcal{Q}$, the LO-QPR problem can be solved in $T\times N=\mathcal{O}(N^{1.5})$ running time, as shown in Theorem \ref{thm:AlgorithmBound}. The outline of quantum learning algorithm is shown in Fig.~\ref{fig:kernel}.

\begin{theorem}
Let quantum kernel $Q(\bm a_{\bm i},\bm x)=|\langle\psi(\bm a_{\bm i})|\psi(\bm x)\rangle|^2$, and $\mathcal{S}=\{(\bm a_{\bm i},b_i)\}_{i=1}^N$ be the training set such that 
\begin{align}
  \Ebb[b_j|\bm a_{\bm j}]=\sum_{i} \alpha_i \abs{\langle\psi(\bm a_{\bm i})|\psi(\bm a_{\bm j})\rangle}^2 + g(\bm a_{\bm j}),  
\end{align}
$g:[0,\infty)\mapsto[-G,G]$ is a strictly monotonic increasing regularisation function such that $\Ebb[g^2]\leq \varepsilon_g$ and $\sum_{ij}\alpha_i\alpha_j \abs{\langle\psi(\bm a_{\bm i})|\psi(\bm a_{\bm j})\rangle}^2<B$. Then for failure probability $\delta\in(0,1)$, $\Ord{N^{5/2}}$ copies of quantum states to estimate $Q(\bm a_{\bm i},\bm x)$, Alg.~\ref{alg:QKernel_alphatron} outputs a hypothesis $\hat{h}^\ast$ such that the training error $\hat{R}_L(\hat{h}^\ast)$ can be bounded by
\begin{align}
\Ord{\sqrt{\varepsilon_g} +  G \sqrt[4]{\frac{\log (1/\delta)}{N}} + B\sqrt{\frac{\log (1/\delta)}{N}} }
\label{eq:risk_qkernel}
\end{align}
 by selecting $\lambda=1$, $T=\mathcal{O}(\sqrt{N/\log(1/\delta)})$ and $M=\mathcal{O}(N\log(T/\delta))$, and 
$R(\hat{h}^*)$ (Eq.~\ref{Eq:generalized_error}) can be upper bounded by
\begin{align}
R(\hat{h}^*)\leq \hat{R}_L(\hat{h}^\ast)+\mathcal{O}\left(\sqrt{\frac{\log(2/\delta)}{N}}\right).
\end{align}
\label{thm:AlgorithmBound}
\end{theorem}

The essential difference between Alg. \ref{alg:QKernel_alphatron} and the original Alphatron algorithm is that we substitute the quantum kernel $Q(\bm{a_i}, \bm{x})$ into the classical kernel function. Although estimating the quantum kernel will introduce an additive error $\epsilon_q$, we rigorously prove the robustness of Alg. \ref{alg:QKernel_alphatron} when encountering measurement errors. \comments{Specifically, utilizing $\Ord{N^{5/2}}$ copies of quantum states to estimate the quantum kernel function, the estimation error $\epsilon_q$ of $Q(\bm{a_i}, \bm{x})$ can be upper bounded by $$\epsilon_q=\Ord{N^{-5/4}\sqrt{\log(1/\delta)}},$$ which provides the quantum learner $\Ord{N^{5/2}}$ quantum joint measurement overhead. This 
gives an upper bound of $\|h^t\pbra{\bm x} - \hat{h}^t\pbra{\bm x}\|$ at the $t$-th iteration step, and the quantum empirical error $R(\hat{h})$ will saturate the upper bound as indicated in Eq.~\ref{eq:risk_qkernel}. }

\emph{\comments{Remark (properties of Alg.~\ref{alg:QKernel_alphatron})}} With the increase of training data scale $N$, Alg.~\ref{alg:QKernel_alphatron} enables $\hat{R}(\hat{h^*})$ of QPR convergence to a low-level empirical risk, which is promised by Theorem~\ref{thm:AlgorithmBound}. Specifically, if the quantum kernel matrix $Q$ can be exactly calculated, then by Goel and Klivans~\cite{goel2019learning}, QKA will output a hypothesis $h^\ast$ with an $\mathcal{O}(\sqrt[4]{\log (1/\delta)/N})$ empirical risk, where $\delta$ represents faliure probability. However, quantum kernel $Q(\bm a_i,\bm x)$ is actually obtained by performing  Destructive-Swap-Test algorithm~\cite{garcia2013swap} finite rounds, and the estimated $\hat{Q}$ has an additive error $\epsilon_q$ to the exact $Q$. In this quantum scenario, utilizing $\Ord{N^{5/2}}$ copies of quantum feature states to estimate the quantum kernel function suffices to provide an estimation of $Q(\bm{a_i}, \bm{x})$ with $\epsilon_q=\Ord{N^{-5/4}\sqrt{\log(1/\delta)}}$ additive error, and this leads the quantum empirical error $\hat{R}(\hat{h^*})$ to saturate $\mathcal{O}(\sqrt[4]{\log (1/\delta)/N})$ risk in $\mathcal{O}(\sqrt{N}t_{Q})$ time complexity, where $t_{Q}$ is the time required to compute kernel function $\hat{Q}$. The above procedure introduces $\Ord{N^{5/2}}$ quantum joint measurements overhead to the quantum learner, nevertheless, it provides a convincing performance for QKA. Further details are explained in the Appendix~\ref{proof_algorithmbound}.

\section{Numerical Simulations}
\label{sec:numericalRes}
\comments{Given a general parameterized Hamiltonian $H(\bm a)$, there only exists specific choices of parameters $\bm a$ for which the ground state $|\psi(\bm a)\rangle$ can be classically solved. Thus the number of collected training data is often limited. In this paper, we explore the capability of QKA algorithm with small-scale of training data in recognizing quantum phases for several instances of LO-QPR tasks. }

Firstly, we consider a warm-up case that detects the appearance of the staggered magnetization for the $S=\frac{1}{2}$XXZ spin chain in the Ising limit~\cite{hieida2001anisotropic}. The Hamiltonian $H_w(g)$ is defined as
\begin{eqnarray}
\begin{split}
\sum\limits_{i=1}^nJ_1\left(S_i^xS_{i+1}^x+S_i^yS_{i+1}^y\right)+J_2S_i^zS_{i+1}^z-g\sum\limits_{i=1}^nS_i^x
\end{split}
\end{eqnarray}
where $S_i^{\alpha}$ is the $\alpha$-component of the $S=1/2$ spin operator at the $i$-th site, and $g$ is the strength of the transverse field. The exchange coupling constant in $xy$ plane is denoted by $J_1$ and that of the $z$-axis direction by $J_2$. Here, we set $J_1=0.2$, $J_2=1$ and depict the phase diagram $M_x=\langle X\rangle$ as a function of $g$ (see the green curve in Fig. \ref{Fig:QPR_demo_1} (c)), where the expectation is under the ground state of Hamiltonian $H_w$. In this case, the number of qubits $n=16$, the training data $\mathcal{S}=\{(g_i, M_x(g_i))\}_{i=1}^N$ where $g_i$ is randomly sampled from the interval $[0,2]$, and the testing data contains all the point from $\{g_t=0.067t\}$ for $0\leq t\leq30$. The predictions proposed by Alg.~\ref{alg:QKernel_alphatron} are illustrated in Fig. \ref{Fig:QPR_demo_1}, \comments{which shows Alg.~\ref{alg:QKernel_alphatron} yields higher accuracy prediction for the training set with more training data. This provides a simulation support for the theoretical bound in Eq.~\ref{eq:risk_qkernel}}

Secondly, we consider a $Z_2\times Z_2$ symmetry-protected topological (SPT) phase $\mathcal{P}$ which contains the $S=1$ Haldane chain. The ground states $\{|\psi(h_1/J,h_2/J)\rangle\}$ belongs to a family of Hamiltonians
\begin{eqnarray}
    \begin{split}
    H_s(h_1/J,h_2/J)=&-J\sum\limits_{i=1}^{n-2}Z_iX_{i+1}Z_{i+2}\\
    &-h_1\sum\limits_{i=1}^nX_i-h_2\sum\limits_{i=1}^{n-1}X_iX_{i+1},
    \label{Eq:HaldaneChain}
\end{split}
\end{eqnarray}
where $X_i$, $Z_i$ are Pauli operators for the spin at site $i$, $n$ is the number of spins, and $h_1$, $h_2$ and $J$ are parameters of $H_s$. In Fig.~\ref{Fig:QPR_demo_2}~(a), the blue and red curves show the phase boundary points, and the background shading (colored tape) represents the phase diagram as a function of $\bm x=(h_1/J,h_2/J)$. When the parameter $h_2=0$, the ground states of $H_s$ can be exactly solvable via the Jordan-Wigner transformation, and it can be efficiently detected by global order parameters whether these ground states belong to the SPT phase $\mathcal{P}$. Here, we utilize $N=40$ data pairs $\{\bm a=(h_1/J,h_2/J), b\}$ as the training data, in which $h_2=0$ and $ b$ indicates phase value on $\bm a$ (see yellow points in Fig. \ref{Fig:QPR_demo_2} (a)). Our target is to identify whether a given, unknown ground state $|\psi(\bm x)\rangle$ belongs to $\mathcal{P}$. In principle, the SPT phase $\mathcal{P}$ can be detected by measuring a non-local order parameter~ \cite{haegeman2012order, pollmann2012detection}
$S_{ij}=Z_{i}X_{i+1}X_{i+3}...X_{j-3}X_{j-1}Z_{j}$, where $X,Z$ are Pauli operators, indices $i<j$ and $i,j\in[n]$, $n$ denotes the number of qubits in concerned Hamiltonian. \comments{Here, we choose $S=Z_1X_1...X_{15}Z_{16}$ and utilize $N^{5/2}=40^{2.5}\approx 10^5$ quantum measurement to estimate the quantum kernel function. The classification results for $n=16$ and $N=15, 40$ are illustrated as Fig.~\ref{Fig:QPR_demo_2}~(b) and (c) respectively, which show the performance on the approximation of the order parameters can be systematically improved by increasing the number of training samples. Fig. \ref{Fig:QPR_demo_2} (c) shows that Alg.~\ref{alg:QKernel_alphatron} can reproduce the phase diagram with high accuracy on $M=4096$ testing points, where $61$ points are mis-classified in the vicinity of paramagnetic boundary, and the classification accuracy $v_s=0.985$ in this case.}

Although the training data is only on the line with $h_2=0$, which can be classically simulated~\cite{cong2019quantum}, a classical learner cannot learn from these data to predict the target quantum phase if the testing data satisfies Conjecture~\ref{con:lowerApprox}. However, we demonstrate that even if the training data can be classically simulated, quantum kernel Alphatron still works. Here, we provide more discussions on this result that shows training in classical data provides a predictive model for quantum points. The reason relies on that the order parameter observable is approximated by a linear combination of feature states in the training set, that is, $\mathcal{M}\approx\sum_i\alpha_i|\psi(\bm a_{\bm i})\rangle\langle\psi(\bm a_{\bm i})|$, and the prediction of the order parameter is significantly determined by the quality of training data. Although the training data are classically simulated, $40$ ground states on line $h_2=0$ suffice to approximate an observable that accurately classifies the quantum phase transition boundary. A quantum learner then can utilize SWAP-test technique to efficiently estimate the quantum kernel $|\langle\psi(\bm x)|\psi(\bm a_i)\rangle|^2$ and predict the quantum phase of $|\psi(\bm x)\rangle$. \comments{In our model, we introduced a regularised term $g(\|h\|)$ to avoid  over-fitting and to enhance its generalization ability. In other words, this model will not necessarily match all training data $\mathcal{S}$, but it has more generalized ability on the testing data $\mathcal{T}$.  As a result, the qualitative domain walls are correctly mapped, but there exists some error in the vicinity of training data. 
} 

We also note that the quantum convolution neural network (QCNN) method~\cite{cong2019quantum} has been proposed to solve the same problem by applying a CNN quantum circuit to the quantum state. Given $n$-qubit ground states $|\psi(\bm x)\rangle$, the QCNN method requires additionally $\mathcal{O}(\frac{7n}{2}(1-3^{1-d})+n3^{1-d})$ multi-qubit operations and $4d$ single-qubit rotations to provide an output
at depth $d$, and the proposed quantum kernel Alphatron requires additionally $\mathcal{O}(n)$ two-qubit operations to estimate the quantum kernel matrix. The measurement complexity of QCNN is determined by the number of iteration steps which is hard to be theoretically analyzed, however the quantum kernel Alphatron only requires $\mathcal{O}(N^2/\epsilon^2)$ quantum measurement in the whole training procedure. The comparison of the computational resources  is summarized in Table~\ref{tab:Compare}. 

\begin{table*}
    \centering
\begin{tabular}{c|cc}
    \hline\hline
   Quantum Resources & QKA & QCNN\\
    \hline
    \multirow{1}{*}{Quantum Gate Complexity} & $\mathcal{O}(n)$ & $\mathcal{O}(\frac{7n}{2}(1-3^{1-d})+n3^{1-d})$\\
    \hline
    \multirow{1}{*}{Quantum Measurement Complexity} & $\mathcal{O}(N^2/\epsilon^2)$, where $\epsilon=N^{-5/4}$  & $\mathcal{O}(dNT/\epsilon^2)$\\
    \hline
    \multirow{1}{*}{The number of iteration Steps $T$} & $\mathcal{O}(\sqrt{N\log(1/\delta)})$ & hard to analyze\\
    \hline\hline
\end{tabular}
    \caption{Quantum computational resources comparison between Quantum Kernel Alphatron (QKA) and Quantum Convolutional Neural Networks (QCNN).}
    \label{tab:Compare}
\end{table*}

Finally, we consider the bond-alternating XXZ model
\begin{align}
    H_b=\sum\limits_{i:\rm{even}}J_{1}H_i + \sum\limits_{i:\rm{odd}}J_{2}H_i,
\end{align}
where $H_i=X_iX_{i+1}+Y_iY_{i+1}+\delta Z_iZ_{i+1}$, and $J_1, J_2, \delta$ are coupling parameters of $H_b$. The XXZ model has three different phases that can be detected by the topological invariant $Z_R(J_1/J_2,\delta)$~\cite{elben2020many}. Here, we select totally $N=60$ pairs $\{\bm a=(J_1/J_2,\delta),\bm b=Z_R(J_1/J_2,\delta)\}$ as the training data on the $\delta=0.5,\delta=3.0$ horizontal lines.
In Fig. \ref{Fig:QPR_demo_3} (b) and Fig. \ref{Fig:QPR_demo_3_1}, we utilize Alg.~\ref{alg:QKernel_alphatron} to generate the phase diagram as a function of $\bm x=(J_1/J_2,\delta)$, where the colored shading background represents the phase classification results on a $16$-qubit system. The data in phase diagram $\mathcal{P}_3$ is post-processed by the averaging scheme. The testing data contains $900$ ground states, where $59$ points are mis-classified in the vicinity of quantum phase transition boundary, and the classification accuracy $v_s=0.934$.

\section{The Power of Learning Algorithms}
\label{sec:complexity}

In this paper, relationships between four different machine learning classes in terms of the method that produces the training data and the learning algorithm are discussed. 
\begin{itemize}
    \item Here, ``Q-Learning Alg.'' (``C-Learning Alg.'') represents learning algorithms with quantum (classical) computer. 
    \item ``Q-Data'' point $(\bm a, b)$ represents the property value $b$ can be observed from physical experiments associated with the system $H(\bm a)$, and ``C-Data'' point $(\bm a, b)$ represents the property $b$ can be efficiently computed by some classical Turing machines given $\bm a$.
\end{itemize}
In this paper, a separation is proved between (C-Learning Alg.~$+$ C-Data) and (Q-Learning Alg.~$+$  Q-Data). Noting that the definition of C-Learning Alg.~$+$ C-Data is different to the `classical' machine learning in Ref.~\cite{huang2021provably}, and the proposed theoretical result does not contradict to their statement.

(1) As shown in the Ref~\cite{huang2021power} (also see Appendix~\ref{App:BPPSamp}), the power of C-Learning Alg.~$+$ C-Data will gradually enhance with the accumulating of training (advice) data, and the set of problems can be solved by classical learning algorithms is defined as the BPP/poly class. With the increase of the training data set, the learner will obtain more and more advicing data, and BPP/poly class will convergence to the P/poly class. It has been proved that the relationship ${\rm BPP}\subseteq{\rm BPP/poly}\subseteq{\rm P/poly}$ holds. Hence, a machine learning task where some data ({even generated classically}) is provided can be considerably different than commonly studied computational tasks. Classical learning algorithm with classical data (C-Learning Alg.~$+$ C-Data) has recently proven successful for many practical applications~\cite{lecun2015deep,zeraati2022flexible, vaswani2017attention, silver2017mastering,jumper2021highly}. However, the direct application of these learning algorithms is challenging for intrinsically quantum problems. This is because the extremely large Hilbert space hinders the efficient translation of many-body problems into a classical learning framework. A natural question is thus raised--where is the limitation of such learning algorithms in quantum problems? Our first result rigorously proved that there exists a LO-QPR problem that cannot be solved by such classical learning algorithms under standard complexity assumptions (see Lemma~\ref{lem:LowerCompute} and Theorem~\ref{thm:LowerLearn}).

{(2) Quantum learning algorithm with quantum data (Q-Learning Alg.~$+$  Q-Data) is another significant theme in this paper. The Q-Learning Alg. can utilize quantum-enhanced feature spaces in the learning process, and ground states are provided by 
quantum circuits with polynomial circuit size.
Here we select variational quantum circuits to provide the quantum-enhanced feature state.
The proposed quantum learner first learns the order parameter from the training data $(\bm a, b)$, then utilizes this approximated order parameter to predict the quantum phase transition. Our second result claims that the LO-QPR problem can be efficiently solved by Q-Learning Alg.~$+$  Q-Data (see Alg.~\ref{alg:QKernel_alphatron} and Theorem~\ref{thm:AlgorithmBound}). These two results thus imply that ``Q-Learning Alg.~$+$  Q-Data'' is strictly stronger than ``C-Learning Alg.~$+$  C-Data'' with suitable complexity assumptions. 
This implies that the quantum kernel Alphatron with quantum data can efficiently solve the LO-QPR problem if there exists a quantum algorithm (or a quantum circuit) that provides the ground state of the concerned Hamiltonians.}

{(3) Another interesting class is classical learning algorithm with quantum data (C-Learning Alg.~$+$  Q-Data). Here, we briefly discuss a special learning scenario for the LO-QPR problem. 
Suppose the training data is provided by a quantum computer, and the quantum linear property is described by a non-local order parameter $\mathcal{M}$. In this setting, the training data has several choices, such as the external parameter and quantum phase $(\bm a, b)$ results from a quantum computer or the classical shadow representation used in \cite{huang2021provably}. Given an unknown ground state $|\psi(\bm x)\rangle$, here we discuss whether a classical learner can predict its quantum phase. Similar to the quantum learning process, a possible follow-up classical learning steps are: (a)  approximate the order parameter $\mathcal{\widetilde{M}}=\sum_iO_i$ from the quantum training data where $O_i$ is the tensor product of Pauli operators; (b) Given the learned $\mathcal{\widetilde{M}}$, estimate the quantum phase value by using pre-obtained random computational basis measurements ($\{0,1\}^n$ bit-string samples) from the quantum state $|\psi(\bm x)\rangle$. Under these two steps, can a classical learner solve the LO-QPR efficiently?
}
{The answer is yet unknown, but we are pessimistic about it. According to Theorem 2 in Ref.~\cite{huang2020predicting}, at least $\Omega(3^{L(\widetilde{\Mcal})}\|\mathcal{\widetilde{M}}\|_{\infty}/\epsilon^2)$ samples are needed to provide an approximation of $\langle\psi(\bm x)|\mathcal{\widetilde{M}}|\psi(\bm x)\rangle$ with additive error $\epsilon$, where $\|\cdot\|_{\infty}$ denotes the spectral norm, and $L(\widetilde{\Mcal})$ is the locality of $\Mcal$. Then in the worst case such that $L(\widetilde{\Mcal}) = n$, it needs an exponential number of samples to obtain the order parameter. This example implies that a classical learner (without a quantum computer) might not efficiently predict a quantum phase transition phenomenon described by non-local order parameters. However, for this non-local order parameter $\Mcal$, we can successfully learn it with Q-Learning Alg.~+ Q-Data. 
}

{(4) The last category is the quantum learning algorithm with classical data (Q-Learning Alg.~$+$  C-Data). Since the relationship $\rm{BPP}\subseteq\rm{BQP}$ holds\footnote{Bounded error Probabilistic Polynomial time (BPP): the class of decision problems solvable by an $\rm NP$ oracle such that: if the answer is `yes' then accept it with at least $2/3$ probability, if the answer is `no' then accept it with at most $1/3$ probability.}, ``Q-Learning Alg.~$+$  C-Data'' could also be strictly stronger than ``C-Learning Alg.~$+$  C-Data''. In this paper, we provide simulation results on this scenario (see Fig.~\ref{Fig:QPR_demo_2}): quantum algorithm learns knowledge from classical data and can predict a classical-hardness quantum phase with high accuracy. Combined with the classical hardness results in Theorem \ref{thm:LowerLearn}, we conjecture that even if the training data are classically simulated, quantum learning algorithm might still have quantum advantages. We leave the rigours proof of this scenario to a future work. }


 We summarize the complexity relationship of these four categories in Fig.~\ref{fig:Problem_class}, where ``Q-Learning Alg.'' refers to the use of a quantum computer, while ``C-Learning Alg.'' relies only on a classical computer; ``Q-Data'' represents learning data directly observed from physical quantum experiments, while ``C-Data'' are  efficiently producible by classical Turing machines.

\begin{figure}[thb]
\centering
  \includegraphics[width=0.50\textwidth]{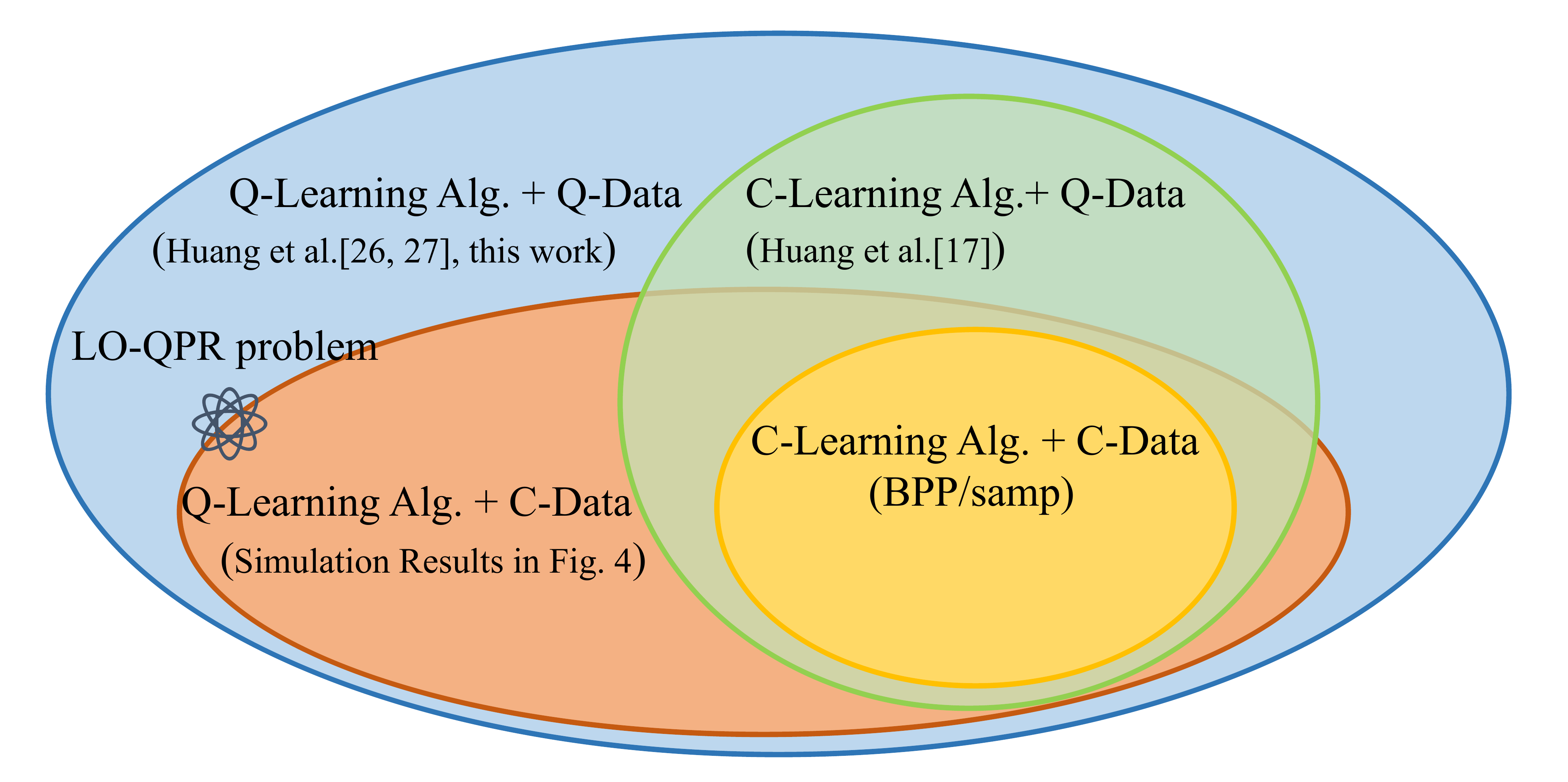}
  \caption{
  Visualization of the learning ability in terms of learning algorithms and data acquiring methods.
  }
  \label{fig:Problem_class}
\end{figure}

\section{Conclusion}
\label{sec:dis}
In this paper, we study the power of classical and quantum learning algorithms in solving LO-QPR problems. Specifically, we prove that under widely accepted assumptions, there exist some LO-QPR problems that cannot be efficiently solved by classical machine learning with classical data. We then prove that LO-QPR problems can be efficiently solved by leveraging the QKA algorithm with quantum data. Furthermore, we provided strong numerical evidence showing that the LO-QPR problems can be solved by the QKA algorithm with quantum data. In some cases, the QKA algorithm succeeded even with only classical data. Based on the above-mentioned theoretical and simulation results, we discussed the complexity relationships of four different machine learning classes in terms of the training data resources and the learning algorithm. We believe the proposed complexity classification helps us understand the power and limitation of classical and quantum learning algorithms.

This work leaves room for further research. For example, our numerical results witnessed the possibility of efficiently solving some LO-QPR problems by QML with classical data, then whether theoretical guarantees exist in showing that the LO-QPR problem belongs to the ``Q-Learning Alg.+C-Data'' class deserves to be further investigated. Finally, exploring the influences of noisy quantum channels on the effectiveness of quantum learning algorithms in solving LO-QPR would be important in practice.

\section*{Author Contributions}
Y. Wu and B. Wu contributed equally to this work. All authors contributed to the discussion of results and
writing of the manuscript.

\section*{Acknowledgement}
We would like to thank Y. Song and P. Rebentrost to provide helpful suggestions for the manuscript. This work is supported by the China Scholarship
Council (Grant No.~202006470011), the High-performance Computing Platform of Peking University, the National Natural Science Foundation of China (Grants No.~12175003, No.~12147133), and Zhejiang Lab's International Talent Fund for Young Professionals.


\appendix

\clearpage
\widetext

\section{Comparison to related works}
Refs.~\cite{huang2022quantum,huang2021information} focused on designing efficient measurement protocols to learn knowledge from an unknown density matrix, then predict its linear property by using accumulated measurement results, which is a learning analogue of shadow tomography problem as shown in Task~\ref{task:dml}.

Authors in~\cite{huang2022quantum,huang2021information} proved that the entangled Bell measurement protocol can efficiently solve Task~\ref{task:dml}, meanwhile, it is classically hard in the \emph{worst-case} scenario in estimating ${\rm Tr}\left(\rho Q\right)$ for some $Q\in\mathcal{P}$. Noting that the claimed quantum advantages might disappear without entangled measurement when $Q$ represents a global observable. In contrast, the proposed quantum advantage in this paper does not depend on the entanglement of multiple copies of quantum states, while the power of quantum-enhanced feature space plays an essential role.

In Task~\ref{def:QPRLearn}, the training data set $\mathcal{S}$ only contains external parameters $\bm a$ and corresponding phase values $b$ rather than the artificially designed density matrix. Furthermore, the order parameter observable $\mathcal{M}$ is unknown, and the learning protocol utilizes Quantum Kernel Alphatron (QKA) to approximate $\mathcal{M}$ by extracting abstract patterns from the data set $\mathcal{S}=\{(\bm a_{\bm i},b_i)\}_{i=1}^N$, then using the approximated $\mathcal{M}$ to construct a prediction model $h(\bm a)$. Then we rigorously proved that there exits a testing set $\mathcal{T}=\{(\bm x_i,y_i)\}_{i=1}^M$ such that $8/9$ of $y_i\in\mathcal{T}$ cannot be efficiently predicted by any classical ML algorithm under some standard complexity assumptions. In the table~\ref{Compare}, we summarize the mainly differences between previous works and this paper. 

\begin{table}[h]
    \centering
\begin{tabular}{c|cc}
    \hline\hline
   Key Properties & This paper & Refs.~\cite{huang2022quantum,huang2021information}\\
    \hline
    \multirow{1}{*}{Training Data} & External parameter $\bm a$ and phase value $b$ & Artificial density matrix $\rho$\\
    \hline
    {Prediction Task} & Given new $\bm a^*\in\mathcal{X}$, predict $b^*$ & Predict ${\rm Tr}\left(Q\rho\right)$\\
    \hline
    \multirow{1}{*}{Related Observable} & Unknown $\mathcal{M}$ & A provided $Q\in\mathcal{P}$\\
    \hline
    \multirow{1}{*}{Classical Hardness Result} & Average-case hardness on $\mathcal{T}$  & Worst-case hardness on $Q$\\
    \hline
    \multirow{1}{*}{Quantum Learning Complexity} & $\mathcal{O}(\sqrt{N\log(1/\delta)}/\epsilon^2)$ & $\mathcal{O}(1/\epsilon^4)$\\
    \hline
    \multirow{1}{*}{Source of Advantage} & Quantum-enhanced feature space & Bell measurement\\
    \hline\hline
\end{tabular}
    \caption{Comparison between this paper and Refs.~\cite{huang2022quantum,huang2021information}.}
    \label{Compare}
\end{table}
From the above comparison, it is clear that previous approaches focus on learning from a single density matrix, but our paper learns patterns from a series of external parameters and their corresponding quantum phases. Then the Bell measurement methods in~\cite{huang2022quantum,huang2021information}
might not be directly applied to the problem studied in this paper.

Recently, Huang et al.~\cite{huang2021provably} utilized an unsupervised learning method to learn samples from provided ground states, which can be summarized as Task~\ref{task:qpl}. Then we show that the required sample complexity $T$ by shadow-tomography based classical ML is expected to increase exponentially with respect to the system size, when the order parameter $\mathcal{M}$ performs on $\mathcal{O}(n)$ qubits. As discussed in~\cite{huang2021provablysup}, only a few LO-QPR problems determined by a global observable have a few-body observable approximation. Therefore, it is reasonable to consider a scenario where the shadow-tomography-formed training data is provided by a quantum computer, and the quantum phase transition can only be determined by a non-local order parameter $\mathcal{M}$. Given an unknown ground state $|\psi(\bm x)\rangle$, here we discuss whether a classical learner can predict its quantum phase. Similar to the quantum learning process, possible follow-up classical learning steps are: 

\begin{itemize}
    \item [(a)]  Approximate the order parameter $\mathcal{\widetilde{M}}=\sum_{i}O_i$ from the quantum training data where $O_i$ is the tensor product of Pauli operators;
     \item [(b)] Given the learned $\mathcal{\widetilde{M}}$, estimate the quantum phase value by using pre-obtained random computational basis measurements ($\{0,1\}^n$ bit-strings) from the quantum state $|\psi(\bm x)\rangle$.
\end{itemize}
 Under these two steps, can a classical learner solve this LO-QPR efficiently? We are pessimistic about it.  According to Theorem 2 in~\cite{huang2020predicting}, at least $\Omega(3^{L(\widetilde{\Mcal})}\|\mathcal{\widetilde{M}}\|_{\infty}/\epsilon^2)$ samples are needed to provide an approximation of $\langle\psi(\bm x)|\mathcal{\widetilde{M}}|\psi(\bm x)\rangle$ with additive error $\epsilon$, where $\|\cdot\|_{\infty}$ denotes the spectral norm, and $L(\widetilde{\Mcal})$ is the locality of $\Mcal$. Then in the worst case such that $L(\widetilde{\Mcal}) = n$, it is expected to require exponential number of samples to solve the LO-QPR problem. This example implies that a classical learner (without a quantum computer) might not efficiently predict a quantum phase transition phenomenon that only can be determined by non-local order parameters.

 While our work utilizes quantum machine learning to extract high-level abstractions from observed data and directly process quantum ground states information by a quantum computer. Here, the ground state $|\psi(\bm a)\rangle$ of $H(\bm a)$ embeds classical external parameter $\bm a$ onto a specific quantum-enhanced feature space, where inner products of such quantum feature states give rise to a quantum kernel, a metric to characterize distances in the feature space. As a result, predicting the ground state property can be transformed into quantum state overlap computation, and thus bypasses the required exponential sample complexity.

Then we emphasize that the definition of ``classical ML'' is different to that in ~\cite{huang2021provably}. In our paper, we discuss the complexity relationship of four categories in terms of the method that produces the training data and the learning algorithm. Here, ``Q-Learning Alg.'' refers to the use of a quantum computer, while ``C-Learning Alg.'' relies only on a classical computer; ``Q-Data'' represents learning data directly observed from physical quantum experiments, while ``C-Data'' are  efficiently producible by classical Turing machines. In this paper, a separation is proved between C-Learning Alg.~$+$ C-Data and Q-Learning Alg.~$+$  Q-Data.

Finally, we point that C-Learning Alg.~$+$ C-Data represents an nontrivial class. As shown in~\cite{huang2021power}, the power of classical learning algorithms will gradually enhance with the accumulating of training (advice) data, and the set of problems can be solved by classical learning algorithms is defined as the BPP/poly class. With the increase of the training data set, the learner will obtain more and more advicing data, and BPP/poly class will be convergence to the P/poly class. It has been proved that the relationship ${\rm BPP}\subseteq{\rm BPP/poly}\subseteq{\rm P/poly}$ holds. Hence, a machine learning task where some  data ({even generated classically}) is provided can be considerably different than commonly studied computational tasks. In our manuscript, we want to demonstrate quantum advantages by introducing quantum computational resources into learning algorithms, and our main contribution is to prove that there exists some LO-QPR problems cannot be efficiently solved by any `C-Learning Alg.~+ C-Data', however, the `Q-Learning Alg.~+ Q-Data' can efficiently solve this learning problem which thus illustrates quantum advantages.

\section{Construct Ground States with variational quantum circuit}
\label{app:archHaar}

\begin{definition}[Architecture]
An architecture $\mathcal{A}$ is a collection of directed acyclic graphs, one for each integer $n$. Each graph consists of $m<{\rm{poly}}(n)$ vertices, and the degree of each vertex $v$ satisfies ${\rm deg}_{in}(v)={\rm deg}_{out}(v)\in\{1,2\}$.
\end{definition}

\begin{definition}[Haar random circuit distribution]
Let $\mathcal{A}$ be an architecture over circuits and let the gates in the architecture be $\{G_i\}_{i=1,...,m}$. Define the distribution $\mathcal{H_A}$ over circuits in $\mathcal{A}$ by drawing each gate $G_i$ independently from the Haar measure.
\end{definition}


We first review the method on constructing a quantum random circuit. As presented in Ref~\cite{boixo2018characterizing}, the quantum random circuit can be constructed in an iterative method in the realistic physical experiment. The construction starts with an initial layer of Hadamard gates to rotate the $X$ basis, and the next $D$ layers alternately insert controlled-Z (${\rm CZ}$) configurations. And one-qubit gates are also randomly sampled from the set $\{{\rm X^{1/2}, Y^{1/2}, T}\}$ and are placed between two ${\rm CZ}$ configurations. Theoretically, the \emph{brickwork} architecture is also can be used to generate quantum random circuits. The brickwork is a kind of structure formed as follows: Perform a string of $2$-qubit gates $U_1\otimes U_2\otimes...\otimes U_{n/2}$ as the first layer, then perform a staggered string of gates, as illustrated in Fig. 2 (a) of the main file. 
\begin{definition}[Haar random quantum circuit]
Let $\mathcal{A}$ be an architecture over circuits and let the gates in the architecture be $\{U_i\}_{i=1,..,R}$. Define the distribution $\mathcal{H_A}$ over circuits in $\mathcal{A}$ by drawing each $2$-qubit gate $U_i$ independently from the Haar measure. Then construct the unitaries along the edges of $\mathcal{A}$, and each constructed circuit is defined as a Haar random quantum circuit.
\end{definition}

In the field of quantum computation, the variational quantum circuit is a popular method for approximating the ground state of  $H(\bm x)$. The key idea of using variational quantum circuit is that the parameterized quantum state $|\Psi(\bm \theta)\rangle$ is prepared and measured on a quantum computer, and the classical optimizer updates the parameters $\bm \theta$ according to the measurement information. The quantum state $|\Psi(\bm \theta)\rangle$ can be prepared by
\begin{align}
|\Psi(\bm \theta)\rangle=U(\bm\theta)|0^n\rangle=\prod\limits_{d=1}^DU_d(\bm\theta_d)|0^n\rangle,
\end{align}
where $U(\bm\theta)$ is composed of $D$ unitaries $U_d(\bm\theta_d)$. Noting that the variational quantum circuit $U(\bm\theta)$ has the same architecture to that of random circuit $U$, and two-qubit gates $U_i(\theta_i)$ are sampled from a subset of $SU(4)$. Then the relationship $\mathcal{U_A}(\bm\theta)\subseteq\mathcal{U_A}$  holds, where $\mathcal{U_A}(\bm\theta)$ and $\mathcal{U_A}$ denote the set of $U(\bm\theta)$ and random quantum circuit $U$ based on $\mathcal{A}$, respectively. 

Then we will show how to utilize one of the instances $U(\bm\theta)\in\mathcal{U_A}(\bm\theta)\subseteq\mathcal{U_A}$ to generate ground states of a family of Hamiltonians $H(\bm x)$~\cite{mcardle2019variational}. The ground state $|\psi(\bm x)\rangle$ of $H(\bm x)$ can be obtained from the imaginary time evolution, that is
\begin{align}
    |\psi(\bm x)\rangle={\rm lim}_{\beta\rightarrow \infty}|\eta(\beta, \bm x)\rangle={\rm lim}_{\beta\rightarrow \infty}A(\beta)e^{-\beta H(\bm x)}|0^n\rangle,
\end{align}
where $\beta$ indicates the inverse temperature, $A(\beta)=1/\sqrt{\langle\phi_0|e^{-2\beta H(\bm x)}|0^n\rangle}$. If we consider the imaginary time evolution of the Schr\"{o}dinger  equation on the variational circuit state space, the parameter dynamics is governed by
\begin{equation}
\sum_{i} \frac{\partial|\Psi(\bm\theta(\beta))\rangle}{\partial \theta_{i}} \dot{\theta_{i}}=-\left(H(\bm x)-E_{\beta}(\bm x)\right)|\Psi(\bm\theta(\beta))\rangle,
\end{equation}
where the term $E_{\beta}(\bm x)=\langle\eta(\beta,\bm x)|H(\bm x)|\eta(\beta, \bm x)\rangle$ and $\bm\theta(\beta)$ denotes the varational parameter in the circuit $U(\bm\theta)$. Applying the McLachlan’s variational principle to minimize the distance between the evolution of variational quantum state $\frac{\partial|\Psi(\bm\theta(\beta))\rangle}{\partial \beta}$ and $-(H(\bm x)-E_{\beta}(\bm x))|\Psi(\bm\theta(\beta))\rangle$, we have
\begin{equation}
\delta \left\|\left(\frac{\partial}{\partial\beta}+H(\bm x)-E_{\beta}(\bm x)\right)|\Psi(\bm\theta(\beta))\rangle\right \|=0,
\end{equation}
and the evolution of parameters $\bm\theta(\beta)$ is obtained from the function
\begin{equation}
\sum_{j} A_{i, j}(\beta) \dot{\theta}_{j}=-C_{i}(\beta),
\end{equation}
in which 
\begin{equation}
\begin{aligned} 
A_{i, j}(\beta) =\Re \left( \frac{\partial\langle\Psi(\bm\theta(\beta))|}{\partial \theta_{i}} \frac{\partial|\Psi(\bm\theta(\beta))\rangle}{\partial \theta_{j}} \right),\\~ C_{i}(\beta) =\Re \left( \frac{\partial\langle\Psi(\bm\theta(\beta))|}{\partial \theta_{i}} H(\bm x)|\Psi(\bm\theta(\beta))\rangle \right).
\end{aligned}
\end{equation}
With the matrix $A(\beta)$ and vector $C(\beta)$, the quantum imaginary time evolution over a small interval $\delta\beta$ can be approximated by solving a linear system $\dot{\bm\theta}(\beta)=A^{-1}(\beta)C(\beta)$, and the variational parameter can be updated by
\begin{align}
    \bm\theta(\beta+\delta\beta)= \bm\theta(\beta)+A^{-1}(\beta)C(\beta)\delta\beta.
\end{align}
Given a large enough $\beta$, and repeat this procedure $N=\beta/\delta(\beta)$ times, the varational quantum state $|\Psi(\bm\theta(\beta))=U(\bm\theta(\beta))|0^n\rangle$ will be an approximation of the ground state $|\psi(\bm x)\rangle$, that is $|\psi(\bm x)\rangle\approx U(\bm\theta(\beta))|0^n\rangle$.
The relationships between unitary sets $\mathcal{U_A}$,  $\mathcal{U_A}(\bm\theta)$ and ground states of $H(\bm x)$ are visualized in Fig.~\ref{fig:random_circuit_2}.

\begin{figure}[t]
\centering
  \includegraphics[width=0.8\textwidth]{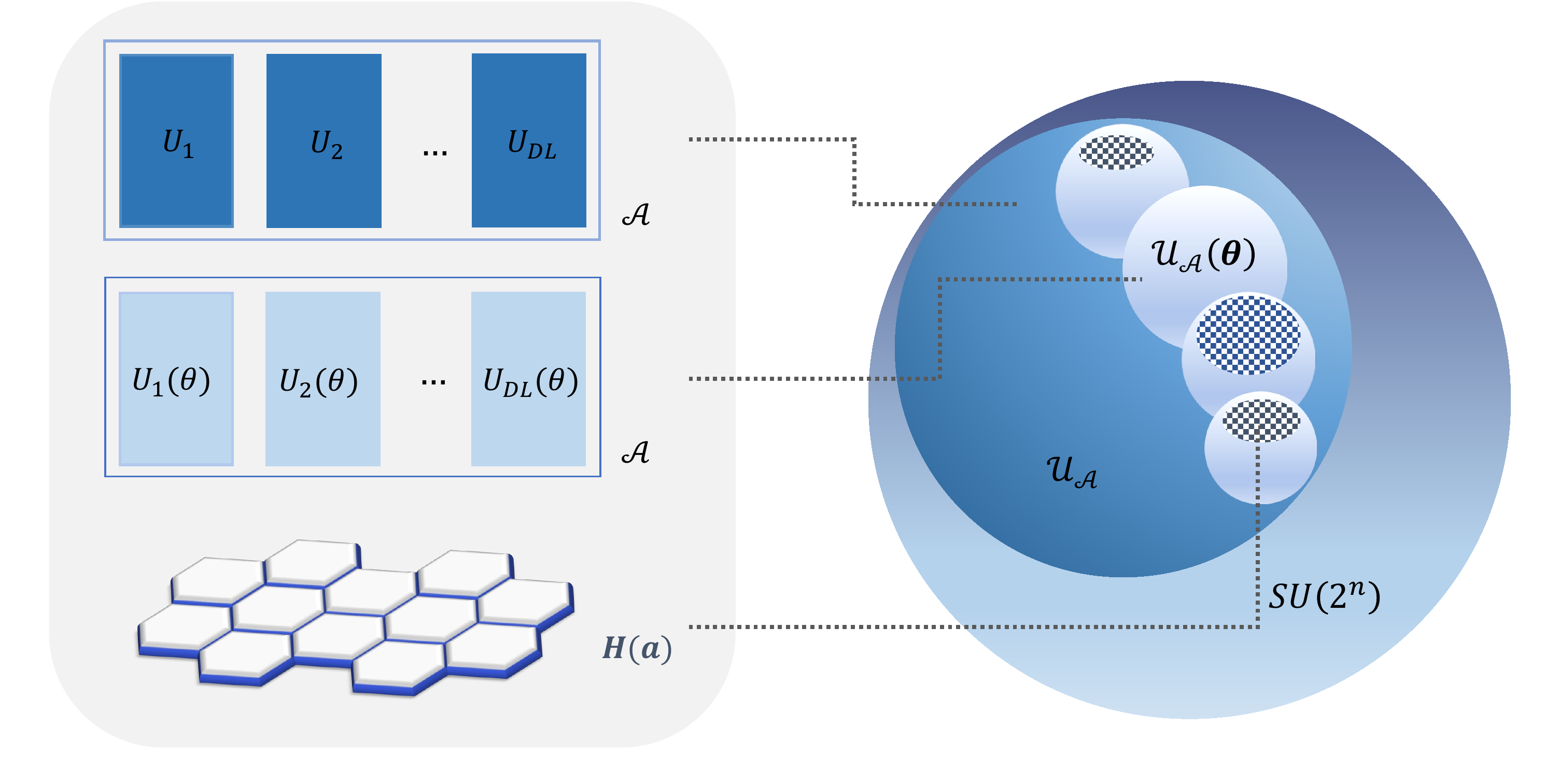}
  \caption{Visualization of the relationships between random quantum states set $\mathcal{U_A}$, variational quantum states set $\mathcal{U_A}(\bm\theta)$ and ground states from a family of Hamiltonian $H(\bm x)$.}
  \label{fig:random_circuit_2}
\end{figure}

\vspace{8pt}

\section{Proof of theorems}
\label{app:ProofThms}
Here, we provide technical details for the proof of theorems in the main text.

\subsection{Proof of Lemma~\ref{lem:LowerCompute}}
\label{proof_lemma_lower_compute}

We first review several lemmas and assumptions which are closely related to our proof.\\
\begin{lemma}[Stockmeyer Theorem~\cite{stockmeyer1985approximation}]
Given as input a function $f:\{0,1\}^n\mapsto\{0,1\}^m$ and any $y\in\{0,1\}^m$ there is a procedure that runs in randomized time ${\rm{poly}}(n,1/\epsilon)$ with access to an ${\rm}{\NP}$ oracle that outputs an $\alpha$ such that 
\begin{align}
    (1-\epsilon)p\le\alpha\le(1+\epsilon)p
\end{align}
for the value
    $$p=\frac{1}{2^n}\sum\limits_xf(x)$$
if the function $f$ can be computed efficiently given $x$.
\end{lemma}
\vspace{5px}
\begin{conjecture}[Ref.~\cite{bouland2019complexity}]
There exists an $n$-qubit quantum circuit $U$ such that the following task is \# P-hard: approximate $p_U(\bm j)=\abs{\langle\bm j | U|0^n\rangle}^2$ to additive error $\epsilon_{c}/2^n$ with probability $\frac{3}{4}+\frac{1}{\poly(n)}$, where $\bm j$ is a $\{0,1\}^n$ bit string and $\epsilon_c=1/{\rm{poly}}(n)$.
\end{conjecture}

Here, a candidate of the worst-case $U\in \Cbb^{2^n\times 2^n}$ is a size $m\leq \poly(n)$ unitary where each basic gate is a two-qubit gate following some fixed gate position architecture $\Acal$. We denote this distribution as $\Hcal_{\Acal}$. Note that the presented conjecture assets that it is $\# P$-hard to compute anything in an interval of radius $1/(2^n{\rm poly}(n))$ around the point $p_U(\bm j)$ on the choice of $U$, however, Bouland et al. proved that it is $\# P$-hard to compute a truncated property $p_{U^{'}} (\bm j)$ which is close to $p_U(\bm j)$ with an exponentially small error. Since this hardness interval is completely contained within the domain of conjectured hardness, their result is necessary for the conjecture. Therefore, if this conjecture holds, it implies computing  most of $p_{U^{'}} (\bm j)$ is $\# P$-hard, and the worst and average-case quantum circuit instances share the same property in the architecture $\mathcal{A}$. 

\vspace{5px}
\begin{proof}[Proof of Lemma~\ref{lem:LowerCompute}] Consider a family of Hamiltonian $\mathcal{H}=\{H(\bm x)\}_{\bm x}$ that is invariant under the Clifford gate, that is $CHC^{\dagger}\in\mathcal{H}$ for any $H\in\mathcal{H}$. 
For a ground state $|\psi(\bm x)\rangle$ of $H(\bm x)\in\mathcal{H}$ satisfies conjecture~\ref{con:lowerApprox}, we can project it to any computational basis $\ket{\bm j}$ with probability $p(\bm j) = \abs{\langle\bm j|\psi(\bm x)\rangle}^2$. The \emph{hiding argument} shows that if one can approximate the probability $p(\bm j)$, then one can approximate $p( 0^n) = \abs{\langle0^n|\psi(\bm x)\rangle}^2$. Therefore Conjecture~\ref{con:lowerApprox} suggests that approximating the $p(\bm j)$ to additive error $2^{-{\rm{poly}}(n)}$ is \# P-hard. 

Then we prove that for $\mathcal{M}\in \mathcal{P}=\{I,X,Y,Z\}^{\otimes n}\setminus I^{\otimes n}$, there exists a ground state $|\psi(\bm y)\rangle$ of $H(\bm y)\in\mathcal{H}$ such that computing $\langle\psi(\bm y)|\mathcal{M}|\psi(\bm y)\rangle$ is classically hard. Consider the observable set $\{\mathcal{M}(\bm s)|\mathcal{M}(\bm s)=Z_1^{s_1}\otimes \cdots \otimes Z_n^{s_n}\}$, where $Z_k$ denotes Pauli-$Z$ operator acts on the $k$-th qubit, and $\bm s=s_1s_2...s_n\in \{0,1\}^{n}$. Then we have
\begin{align}
    o_{\bm s}=\langle\psi(\bm x)|\mathcal{M}(\bm s)|\psi(\bm x)\rangle=\sum_{\bm j}p(\bm j)(-1)^{\bm j\cdot\bm s},
\end{align}
and $o_{\bm s}/2^n$ is the Fourier transformation of $p(\bm j)$. Based on the algebra symmetry between $p(\bm j)$ and $o_{\bm s}/2^n$, we have
\begin{align}
p(\bm j)=\sum_{\bm s}o_{\bm s}(-1)^{\bm j\cdot\bm s}/2^n.
\end{align}
If $o_{\bm s}$ can be efficiently approximated by a classical computer given $\bm s$, there exists a $\rm{BPP^{NP^{BPP}}}$ algorithm that can approximate $p(\bm j)$ with the multiplicative error $1/{\rm{poly}}(n)$ based on a theorem by Stockmeyer~\cite{stockmeyer1985approximation}. Considering $\rm{BPP}\subseteq P/poly$ and approximating $p(\bm j)$ is $\# P$-hard, these yield $\rm{P^{\# P}}\subseteq \rm{BPP^{NP^{BPP}}}\subseteq\rm{BPP^{NP}}/poly$. Since $\rm{NP^{NP}}\subseteq P^{\# P}$, one has $\rm{NP^{NP}}\subseteq \rm{BPP^{NP}}/poly$,
which implies PH collapses to the second level~\cite{arora2009computational}.

Therefore, there does not exist a classical algorithm that can efficiently calculate $o_{\bm s^*}$ for some $\bm s^*\in\{0,1\}^n$ based on the assumption that PH does not collapse and Conjecture~\ref{con:lowerApprox} holds. Without loss of generality, let $\mathcal{M}(s^*)=Z_1^{s_1^*}\otimes Z_2^{s_2^*}\otimes \cdots \otimes Z_n^{s_n^*}$, and the interested physical order parameter observable $\mathcal{M}_t=P_1\otimes\cdots\otimes P_n$ (for example, ferromagnetic parameter $X$ or SPT parameter $Z_{i}X_{i+1}X_{i+3}...X_{j-3}X_{j-1}Z_{j}$). Since a Clifford gate $C$ maps a Pauli operator to another Pauli operator, then the target order parameter observable $\mathcal{M}_t$ can be expressed as $\mathcal{M}_t=\left(C\right)^{\dagger}Z_1^{s_1^*}\otimes Z_2^{s_2^*}\otimes \cdots \otimes Z_n^{s_n^*}\left(C\right)$, where $C$ represents a Clifford gate. Therefore, for any target $n$-qubit Pauli observable $\mathcal{M}_t$, the expectation value 
\begin{align}
\langle\psi(\bm x)|\mathcal{M}(\bm s^*)|\psi(\bm x)\rangle=\langle\psi(\bm x)|C\left(C^{\dagger}Z_1^{s_1^*}\otimes \cdots \otimes Z_n^{s_n^*}C\right)C^{\dagger}|\psi(\bm x)\rangle=\langle\psi(\bm y)|\mathcal{M}_t|\psi(\bm y)\rangle.
\end{align}
The last equality is valid because $CH(\bm x)C^{\dagger}$ belongs to the Hamiltonian family $\mathcal{H}$. If $|\psi(\bm x)\rangle$ represents the ground state of $H(\bm x)$, then $C|\psi(\bm x)\rangle$ will be the ground state of $CH(\bm x)C^{\dagger}$.
Therefore, for any interested order parameter observable $\mathcal{M}_t\in\mathcal{P}$, there exists a ground state $|\psi(\bm y)\rangle\in\mathcal{H}$ such that their corresponding quantum phase is classically hard.
\end{proof}

\subsection{Complexity argument for the power of data}
\label{App:BPPSamp}
Here, we review the power of classical ML algorithms that can learn from data by means of a complexity class, which is defined as ${\rm{BPP/poly}}$ in Ref.~\cite{huang2021power}. A language $L$ of bit strings is in ${\rm{BPP/poly}}$ if and only if the following holds. Suppose $M$ and $D$ are two probabilistic Turing machines, where $D$ generates samples $\bm x$ with $|\bm x|=n$ in polynomial time for any size $n$ and $D$ defines a sequence of input distributions $\{D_n\}$. $M$ takes an input $\bm x$ of size $n$ along with a set $\{(\bm x_i,y_i)\}_{i=1}^{{\rm{poly}}(n)}$, where $\bm x_i$ is sampled from $D_n$ using $D$ and $y_i$ indicates the corresponding label. If $\bm x_i\in L$, one has $y_i=1$, else $y_i=0$. Specifically, one requires:\\
(1) The probabilistic Turing machine $M$ processes all inputs $\bm x$ in polynomial time.\\
(2) For all $\bm x\in L$, $M$ outputs $1$ with probability greater than $2/3$.\\
(3) For all $\bm x\notin L$, $M$ outputs $0$ with probability less than $1/3$. 

From the above definition, we know that BPP is contained in this complexity class. Now we provide details on the separation between classical ML algorithms with classical data and BPP. Consider an undecidable language $L_{h}=\{1^n|n\in A\}$, where $A$ is a subset of the natural numbers set, and consider a classically easy language $L_e\in {\rm BPP}$. Assuming that for any input size $n$, there exists an input $a_n\in L_e$ and an input $b_n\notin L_e$. Then a new language can be defined:
\begin{align}
    L=\bigcup\limits_{n=1}^{\infty}\{x|\forall x\in L_e, 1^n\in L_h, |x|=n\}\cup \{x|\forall x\notin L_e, 1^n\notin L_h, |x|=n\}.
\end{align}
For each size $n$, if $1^n\in L_h$, the language $L$ would include all $x\in L_e$ with $|x|=n$, otherwise, the language $L$ would include all $x\notin L_e$ with $|x|=n$. That is to say, if one can decide whether a problem $x\in L$ for an input $x$ using a classical algorithm, we can output whether $1^n\in L_h$ by checking whether $x\in L_e$. This is impossible since the language $L_h$ is undicidable. Hence the language $L$ is not in BPP class. On the other hand, if the training data $\{x_i,y_i\}$ are provided, where the label $y_i$ represents whether $x_i$ belongs to $L$, and we thus can decide whether $1^n$ belongs to $L_h$.

Based on the above discussion, we know that the power of classical learning algorithms will gradually enhance with the accumulating of training (advice) data, and the set of problems can be solved by classical learning algorithms is defined as the BPP/poly class. With the increase of the training data set, the learner will obtain more and more advicing data, and BPP/poly class will be convergence to the P/poly class. Hence, a machine learning task where some data is provided can be considerably different than commonly studied computational tasks. In our manuscript, we want to demonstrate quantum advantages by introducing quantum computational resources into learning algorithms, and our main contribution is to rigorously prove that any `C-Learning Algorithm + C-Data' cannot solve the quantum phase learning problem. However, the `Q-Learning Algorithm + Q-Data' can efficiently solve this learning problem which thus illustrates quantum advantages.

\subsection{Proof of Theorem~\ref{thm:LowerLearn}}
\label{subapp:LowerLearn}
\noindent The following lemma gives an average-case hardness for the quantum phase computation problem.

\begin{lemma}
With the assumption that Conjecture~1 holds, and the PH in the computational complexity theory does not collapse, it is classically hard to approximate 8/9 of the quantum phase computation problem given a certain $n$-qubit Hamiltonian $H(\bm a)$ with additive error $\epsilon = 1/(\poly(n))$, where its ground state 
$\ket{\psi(\bm a) }= U(\vec{\bm \theta}(\bm a))|0^n\rangle$ and 
$U(\vec{\bm\theta}(\bm a))\in \Ucal_\Acal(\bm \theta)$.
\label{lem:Lower_average}
\end{lemma}

\begin{proof}
{

Suppose we take a worst-case ground state $|\Psi\rangle=U(\vec{\bm\theta})|0^n\rangle$ of $H(\bm a)$ generated by a variational quantum circuit $U(\vec{\bm\theta})\in\mathcal{U_A}(\bm\theta)$, such that computing $p(\bm j)=|\langle \bm j|\Psi\rangle|^2$ to within additive error $2^{-{\rm{poly}}(n)}$ is $\# P$-hard (based on conjecture~1 in the main file). Since the two-qubit gate structures of $U(\vec{\bm\theta})=U_{DL}(\bm\theta_{DL})\cdots U_{1}(\bm\theta_{1})$ is provided in $\mathcal{A}$, where $U_r(\bm\theta_r)$ denotes the $r$-th 
two-qubit gate for $r\in[DL]$, $\vec{\bm\theta}=(\bm\theta_1,\dots,\bm\theta_{DL})$ and $\bm\theta_r\in\mathbb{R}^{15}$. Denote $R=DL$, each two-qubit gate
\begin{align}
    U(\bm\theta_r)=\exp\left(-i\sum\limits_{j_1,j_2=0}^4\theta_r(j_1,j_2)\left(P_{j_1}\otimes P_{j_2}\right)\right)=\exp\left(-i\langle\bm\theta_{r},\bm P_r\rangle\right),
\end{align}
where $P_{j}\in\{I,X,Y,Z\}$ and each $\theta_r(j_1,j_2)\in[0,2\pi]$. 
Using Taylor series, one obtains
\begin{align}
U(\vec{\bm\theta})=\prod\limits_{r=1}^{R}\sum\limits_{k=0}^{\infty}\frac{(-i\langle\bm\theta_{r},\bm P_r\rangle)^k}{k!}.
\end{align}
Denote 
\begin{align}
    U(\bm\theta_r)_{\rm tr}=\sum\limits_{k=0}^{K}\frac{(-i\langle\bm\theta_{r},\bm P_r\rangle)^k}{k!},
\end{align}
therefore $U(\bm\theta_r)-U(\bm\theta_r)_{\rm tr}=\sum_{k=K+1}^{\infty}\frac{(-i\langle\bm\theta_{r},\bm P_r\rangle)^k}{k!}$. For arbitrary bit-string $x,y$, we can apply standard bound on Taylor series to bound $\|\langle x|U(\bm\theta_r)-U(\bm\theta_r)_{\rm tr}|y\rangle\|_1\leq\kappa/K!$ for some constant $\kappa$.
Therefore for an arbitrary observable $\mathcal{M}$
\begin{eqnarray}
\begin{split}
    &\langle0^n|U^{\dagger}(\vec{\bm\theta})\mathcal{M} U(\vec{\bm\theta})|0^n\rangle=\sum\limits_{i,j=0}^{2^n-1}\mathcal{M}_{ij}\langle0^n|U^{\dagger}(\vec{\bm\theta})|i\rangle\langle j|U(\vec{\bm\theta})|0^n\rangle\\
    &=\sum\limits_{i,j=0}^{2^n-1}\mathcal{M}_{ij}\left(\sum\limits_{\substack{y_1,y_2,...y_{R-1}\in\{0,1\}^n\\y_{R}=i}}\prod\limits_{r=1}^{R}\langle0^n|U(\bm\theta_r)|y_r\rangle\right)\left(\sum\limits_{\substack{y_1,y_2,...y_{R-1}\in\{0,1\}^n\\y_{R}=j}}\prod\limits_{r=1}^{R}\langle y_r|U(\bm\theta_r)|0^n\rangle\right)\\
    &=\sum\limits_{i,j=0}^{2^n-1}\mathcal{M}_{ij}\left(\sum\limits_{\substack{y_1,y_2,...y_{R-1}\in\{0,1\}^n\\y_{R}=i}}\prod\limits_{r=1}^{R}\langle0^n|\sum\limits_{k=0}^{\infty}\frac{(-i\langle\bm\theta_{r},\bm P_r\rangle)^k}{k!}|y_r\rangle\right)\left(\sum\limits_{\substack{y_1,y_2,...y_{R-1}\in\{0,1\}^n\\y_{R}=j}}|\prod\limits_{r=1}^{R}\langle y_r|\sum\limits_{k=0}^{\infty}\frac{(-i\langle\bm\theta_{r},\bm P_r\rangle)^k}{k!}|0^n\rangle\right),
\end{split}
\end{eqnarray}
where the $y_1,.y_2,...$ represent Feymann integration path. Since $\langle y_r|U(\bm\theta_r)|0^n\rangle$ can be approximated by a polynomial of degree $K$ based on Taylor truncated method, the above expression can be rewritten by
\begin{eqnarray}
    \sum\limits_{r,s=0}^{2^n-1}\mathcal{M}_{rs}\left(f_r(\bm\theta_1,...\bm\theta_R)+\mathcal{O}\left(\frac{2^{Rn}}{(K!)^{R}}\right)\right)\left(f_s(\bm\theta_1,...\bm\theta_R)+\mathcal{O}\left(\frac{2^{Rn}}{(K!)^{R}}\right)\right)
    \label{Eq:mean}
\end{eqnarray}
where $f_r$ represents a polynomial of degree $RK$. Furthermore, $f_r(\bm\theta_1,...,\bm\theta_R)$ can be approximated by a low-degree function with at most $C_R^q(K)^q$ terms, where $q=\mathcal{O}(1)$ and $C_R^q$ represents the combination number. Let $f_r(\vec{\bm\theta})=\sum_{\vec{\bm i}}\bm\alpha_{\vec{\bm i}}\bm\theta_1^{i_1}\cdots \bm\theta_R^{i_R}$, where $\vec{\bm i}=(i_1,...,i_R)$ and each $i_l\in[K]$, $l\in[R]$. For every term $\bm\alpha_{\vec{\bm i}}\bm\theta_1^{i_1}\cdots \bm\theta_v^{i_v}$ with $i_1=\cdots=i_v=K$ and $v>q$, its corresponding parameter $\abs{\bm\alpha_{\vec{\bm i}}}\leq 1/(K!)^{q}$ (based on Taylor series). Therefore, the relationship
\begin{align}
\Delta f_r(\vec{\bm\theta})=\abs{f_r-\tilde{f}_r}=2^{Rn}\abs{\sum_{\vec{\bm i}}\bm\alpha_{\vec{\bm i}}\bm\theta_1^{i_1}\cdots \bm\theta_R^{i_R}-\sum_{\substack{j_1,...,j_q\leq K-1}}\bm\alpha_{\vec{\bm j}}\bm\theta_{j_1}^{j_1}\cdots \bm\theta_{j_q}^{j_q}}\leq 2^{Rn}\left(\frac{K^R}{(K!)^q}\right)
\end{align}
holds. Then let $q=\mathcal{O}(1)$, $R=\mathcal{O}(n^2)$, $K={\rm poly}(n)$ and $K\gg R$,  $\tilde{f}_r$ can provide an estimation to $f_r$ within $2^{-{\rm poly}(n)}$ additive error, and $\tilde{f}_r$ only has $C_R^q(K)^q={\rm poly}(n)$ terms. Then Eq.~\ref{Eq:mean} can be represented by a muti-variable polynomial function $f(\vec{\bm\theta},\mathcal{M})$ with $R$ variables and at most ${\rm poly}(n)$ terms, and the relationship
\begin{align}
    \|\langle0^n|U^{\dagger}(\vec{\bm\theta})\mathcal{M} U(\vec{\bm\theta})|0^n\rangle-f(\vec{\bm\theta},\mathcal{M})\|\leq 2^{-{\rm poly}(n)}
\label{Eq:fit}
\end{align}
holds.


Suppose the variational quantum circuit $U(\vec{\bm\theta})$ is powerful enough, such that it can calculate some ground states $\mathcal{G}=\{|\psi(\bm a_i)\rangle=U(\vec{\bm\theta}(\bm a_i))|0^n\rangle\}_{i=1}^M$ of the Hamiltonian $H(\bm a)$. The `worst-to-average-case' reduction can be achieved by \emph{proof of contradiction}:


Since the variational quantum circuit can generate a ground state set $\mathcal{G}$ for a family of Hamiltonian $H(\bm a)$, suppose there exists a classical algorithm $\mathbf{\mathcal{O}}$, which can efficiently approximate $8/9$ of $\{\langle0^n|U^{\dagger}(\vec{\bm\theta}(\bm a_i))\mathcal{M} U(\vec{\bm\theta}(\bm a_i))|0^n\rangle\}_{i=1}^M$, where $M=100C_R^q(K)^q$.
It implies that for at least 2/3 choices of $\{U(\vec{\bm\theta}(\bm a_i))\}$, $\mathbf{\mathcal{O}}$ correctly approximate $\{b_i=\langle0^n|U^{\dagger}(\vec{\bm\theta}(\bm a_i))\mathcal{M} U(\vec{\bm\theta}(\bm a_i))|0^n\rangle\}$. According to the assumption in $\mathcal{G}$, the variational quantum circuit provides a map between $\bm a_i\mapsto\vec{\bm\theta}(\bm a_i)$.
From Eq.~\ref{Eq:fit}, one can fit a polynomial function in $\bm\theta$ that recovers the value of $\langle0^n|U^{\dagger}(\bm\theta)\mathcal{M} U(\bm\theta)|0^n\rangle$ by using $\{\vec{\bm\theta}(\bm a_i)\}$. However, according to the Lemma~\ref{lem:LowerCompute}, successful approximating $\langle0^n|U^{\dagger}(\bm\theta)\mathcal{M} U(\bm\theta)|0^n\rangle$ (worst-case scenario) by a $\rm BPP$ algorithm will yield PH collapse. Then, it is hard to approximate 8/9 of the $\{b_i\}$.  Then the above $(\vec{\bm\theta}(\bm a_i),b_i)$ can be used in constructing a testing set
\begin{align}
\mathcal{T}=\{(\bm x_i,y_i)||\psi(\bm x_i)\rangle\in\mathcal{G},y_i=\langle\psi(\bm x_i)|\mathcal{M}|\psi(\bm x_i)\rangle\},
\end{align}
where $|\psi(\bm x_i)\rangle=U(\vec{\bm\theta}(\bm a_i))|0^n\rangle$ and $y_i=b_i$.
}

\end{proof}

\noindent \emph{Note.} One might think that the above procedure could inspire a classical learning algorithm in predicting a hard quantum phase by using quantum data, however this cannot be directly used in solving LO-QPR problems. The reason is that the above procedure fits a polynomial function in $\bm\theta$ rather than the external parameter $\bm a$ in the training data set, then the proof is not sufficient in proving the efficiency of C-Learning Alg.+Q-Data on LO-QPRs determined by a global linear observable.

\begin{proof}[Proof of Theorem~\ref{thm:LowerLearn}]
Our proof depends on the quantum circuit representation of concerned ground states. We firstly provide the method on constructing a
test set $\mathcal{T}=\{\bm x_{\bm i},y_i\}_{i=1}^M$. According to the construction in Sec.~\ref{app:archHaar}, the variational quantum circuit state $|\Psi(\vec{\bm\theta}(\bm x))\rangle$ can approximate a ground state $|\psi(\bm x)\rangle$ of the Hamiltonian $H(\bm x)$. Starting from $|\Psi(\vec{\bm\theta}(\bm x))\rangle$, we want to approximate the ground state $|\psi(\bm x_i+\delta\bm x)\rangle$ of $H(\bm x_i+\delta\bm x)$ at external parameter $\bm x_i+\delta\bm x$ by $|\Psi(\vec{\bm\theta}(\bm x)+\delta\bm{\theta})\rangle$. The value of $\delta\bm{\theta}=(\delta\bm{\theta}_1,\delta\bm{\theta}_2,...,\delta\bm{\theta}_{DL})$ can be determined by minimizing the distance
\begin{align}
\mathcal{L}(\delta\bm{\theta})=\|\mathbf{d}|\psi(\bm x_i+\delta\bm x)\rangle-\mathbf{d}|\Psi(\vec{\bm\theta}(\bm x)+\delta\bm{\theta})\rangle\|,
\label{Eq:difference}
\end{align}
where
\begin{align}
\mathbf{d}|\psi(\bm x_i+\delta\bm x)\rangle=|\psi(\bm x_i+\delta\bm x)\rangle-|\psi(\bm x_i)\rangle
, \end{align}
and
\begin{align}
\mathbf{d}|\Psi(\vec{\bm\theta}(\bm x)+\delta\bm{\theta})\rangle=\sum\limits_{d=1}^{DL}\frac{\partial|\Psi(\vec{\bm\theta}(\bm x))\rangle}{\partial\bm{\theta}_d}\delta\bm{\theta}_d,
\end{align}
and the notation $\|\cdot\|$ represents the fidelity norm. Then the function $\mathcal{L}^2(\delta\bm{\theta})$ can be further computed as
\begin{eqnarray}
\begin{split}
&\mathbf{d}\langle\psi(\bm x_i+\delta\bm x)|\mathbf{d}|\psi(\bm x_i+\delta\bm x)\rangle-\sum\limits_{d=1}^{DL}\langle\psi(\bm x_i+\delta\bm x)|\frac{\partial|\Psi(\vec{\bm\theta}(\bm x))\rangle}{\partial\bm{\theta}_d}\delta\bm{\theta}_d\\
&-\sum\limits_{d=1}^{DL}\frac{\partial\langle\Psi(\vec{\bm\theta}(\bm x))|}{\partial\bm{\theta}_d}|\psi(\bm x_i+\delta\bm x)\rangle\delta\bm{\theta}_d+\sum\limits_{m,s}\frac{\partial\langle\Psi(\vec{\bm\theta}(\bm x))|}{\partial\bm{\theta}_m}\frac{\partial|\Psi(\vec{\bm\theta}(\bm x))\rangle}{\partial\bm{\theta}_s}\delta\bm{\theta}_m\delta\bm{\theta}_s.
\end{split}
\end{eqnarray}	
If we focus on the $m$-th variable $\delta\bm{\theta}_m$, the minimum of $\mathcal{L}^2(\delta\bm{\theta})$ obtains at
\begin{align}
\sum\limits_{s=1}^{DL}B_{s,m}\delta\bm{\theta}_m=E_m,
\end{align}
in which the parameter
\begin{align}
B_{s,m}=\Re\left(\frac{\partial\langle\Psi(\vec{\bm\theta}(\bm x))|}{\partial\bm{\theta}_s}\frac{\partial|\Psi(\vec{\bm\theta}(\bm x))\rangle}{\partial\bm{\theta}_m}\right),
\end{align}
and
\begin{align}
E_m=\Re\left(\frac{\partial\langle\Psi(\vec{\bm\theta}(\bm x))|}{\partial\bm{\theta}_m}\mathbf{d}|\psi(\bm x_i+\delta\bm x)\rangle\right).
\end{align}
Once each elements are estimated, the variation of parameters $\delta\bm{\theta}$ can be efficiently computed by solving the linear system
\begin{align}
B(\bm{\theta})\delta\bm{\theta}=E(\bm{\theta}),
\end{align}
where the matrix $B(\vec{\bm\theta}(\bm x))=(B_{s,m})_{DL\times DL}$ and $E(\vec{\bm\theta}(\bm x))=(E_1,...,E_{DL})^T$. Since the matrix $B$ is a real-valued symmetry matrix, the inverse of $B$ must exist. And $\vec{\bm\theta}(\bm x)$ can be updated by
\begin{align}
\vec{\bm\theta}(\bm x)+\delta\bm{\theta}=\vec{\bm\theta}(\bm x)+B^{-1}(\vec{\bm\theta}(\bm x))E(\vec{\bm\theta}(\bm x)).
\end{align}
Then, the ground state $|\psi(\bm x_i+\delta\bm x)\rangle$ can be approximated by $|\Psi(\vec{\bm\theta}(\bm x)+\delta\bm{\theta})\rangle$. In this iterative method, one can construct a series of $(|\Psi(\vec{\bm\theta}(\bm x))\rangle, |\Psi(\vec{\bm\theta}(\bm x)+\delta\bm\theta)\rangle,...)$ to represent the ground state $|\psi(\bm x)\rangle, |\psi(\bm x+\delta\bm x)\rangle,...$ from a family of Hamiltonian $H(\bm x)$, and this thus constructs a testing set $\mathcal{T}$ based on the Hamiltonian $H(\bm x)$ and the architecture $\mathcal{A}$.

Now we only need to prove that there does not exist efficient classical ML algorithm that can predict $y_i$ for $\bm x_{\bm i}\in\mathcal{T}$ (constructed in Lemma~\ref{lem:Lower_average}) with probability 8/9. The basic idea relies on: if the classical ML can predict all $y_i\in\mathcal{T}$, then we can design an efficient classical algorithm to solve the worst-case hardness GLP problem, which results in a contradiction. Given the classical training set $\mathcal{S}$ (C-Data), the power of classical ML can be characterized as the $\rm{BPP/samp}$ class~\cite{huang2021power}.  Suppose there exists a classical ML, which can efficiently predict $8/9$ of $\{y_i=\langle\psi(\bm x_i)|\mathcal{M}|\psi(\bm x_i)\rangle\}_{i=1}^M$, where $M=100KR$. Then given $\{(\vec{\bm\theta}(\bm x_i))\}_{i=1}^M$ and the vector $\vec{\bm\theta}(\bm x)$ (parameter in the worst-case scenario), one can fit a polynomial function in $\vec{\bm\theta}$ that recovers the value of $\langle0^n|U^{\dagger}(\vec{\bm\theta}(\bm x))\mathcal{M} U(\vec{\bm\theta}(\bm x))|0^n\rangle$ which is the worst-case scenario. According to Lemma~1, an algorithm $O$ that can approximate the worst-case scenario $\langle0^n|U^{\dagger}(\vec{\bm\theta})\mathcal{M} U(\vec{\bm\theta})|0^n\rangle$ with $1/{\rm poly}(n)$ additive error implies a $\rm{BPP^{NP^{O}}}$ algorithm can approximate $p(\bm j)=\abs{\langle\bm j|\Psi\rangle}^2=|\langle\bm j|U(\vec{\bm\theta})|0^n\rangle|^2$ with the multiplicative error $1/{\rm{poly}}(n)$ based on a theorem by Stockmeyer~\cite{stockmeyer1985approximation}. Therefore, if there exists a classical ML with classical data can efficiently predict $8/9$ of $y_i\in\mathcal{T}$, this implies a $\rm{BPP^{NP^{BPP/samp}}}$ algorithm that can approximate $p(\bm j)$ with the multiplicative error. Considering $\rm{BPP/samp}\subseteq P/poly$ and approximating $p(\bm j)$ is $\# P$-hard, these yield $$\rm{P^{\# P}}\subseteq \rm{BPP^{NP^{BPP/samp}}}\subseteq\rm{BPP^{NP}}/poly.$$ Since $\rm{NP^{NP}}\subseteq P^{\# P}$, one has $\rm{NP^{NP}}\subseteq \rm{BPP^{NP}}/poly$, which implies PH collapses to the second level~\cite{arora2009computational}. Hence with the assumption that PH does not collapse, classical machine learning with classical resources cannot solve LO-QPRs even in the average-case scenario on $\mathcal{T}$.
\label{proof_them_1}
\end{proof}

\subsection{Proof of Theorem~\ref{thm:AlgorithmBound}}
\label{proof_algorithmbound}
\begin{proof}[Proof sketch of Theorem \ref{thm:AlgorithmBound}]
Let $\bm w=\sum_{i=1}^N\alpha_i|\psi(\bm a_{\bm i})\rangle\otimes|\psi(\bm a_{\bm i})\rangle^\ast$, where $\ket{\psi(\cdot)}^\ast$ is the conjugate of $\ket{\psi(\cdot)}$, and the reproduced-kernel-feature-vector $\widetilde{\Psi}(\bm a_{\bm i})=|\psi(\bm a_{\bm i})\rangle\otimes|\psi(\bm a_{\bm i})\rangle^\ast$. 
Then
\begin{align}
    Q(\bm a_{\bm i},\bm x)&=|\langle\psi(\bm a_{\bm i})|\psi(\bm x)\rangle|^2=\langle \widetilde{\Psi}(\bm {a_i})|\widetilde{\Psi}(\bm x)\rangle,
\end{align}
and
\begin{align}
\sum_{i} \alpha_i \abs{\langle\psi(\bm a_{\bm i})|\psi(\bm x)\rangle}^2 =\abra{\bm w, \widetilde{\Psi}(\bm x)},
\end{align}
by the definition of $\bm w$ and $\ket{\Psi(\bm a_{\bm i})}$. Therefore, $\Ebb\sbra{b_j|\bm{a_j}} = \abra{\bm w, \Psi(\bm x)} + g\pbra{\bm{a_j}}$ and $$\vabs{\bm w}^2 = \sum_{ij}\alpha_i\alpha_j \abs{\langle\psi(\bm {a_i})|\psi(\bm{a_j})\rangle}^2< B.$$
Hence, this theorem is followed by substituting the quantum kernel $Q$ and feature map $\widetilde{\Psi}$ into Theorem 1 of Goel and Klivans~\cite{goel2019learning}, if we can implement $Q$ perfectly. Nevertheless, we can only approximate it with small additive error via quantum circuit. Specifically,
the quantum kernel can be approximated by independently performing the Destructive-Swap-Test~\cite{garcia2013swap} to $\mathcal{O}(\log (1/\delta)/\varepsilon^2)$ copies of $2n$-qubit state $\ket{\psi\pbra{\bm {a_i}}}\otimes \ket{\psi\pbra{\bm x}}$, with additive error $\epsilon_Q$ and failure probability $\delta$, see SI material for the details of the circuit implementation of quantum kernel. If the quantum kernel $Q(\bm a_i,\bm x)$ is estimated by performing Destructive-Swap-Test algorithm $\Ord{N^{5/2}}$ times, then
 \begin{align}
     \abs{h^t\pbra{\bm x} - \hat{h}^t\pbra{\bm x}} \leq \Ord{t^2\sqrt{\log\pbra{1/\delta}}/N^{5/4}}
  \label{pointed_bound}
 \end{align}
holds with $1-\delta$ probability, where $h^t\pbra{\bm x}$, $\hat{h}^t\pbra{\bm x}$  represent the ideal QKA model and estimated QKA model at the $t$-th iteration step.
\end{proof}

\comments{The anticipated range of Eq.~\ref{pointed_bound} matches perfectly to our simulation results, which can be checked from the generalized risk $\hat{R}_L(\hat{h}^\ast)$ at $t=T$, as demonstrated in Fig.~4 (b) and Supplementary material. For example, the theoretical upper bound in detecting SPT is $\hat{R}_L(\hat{h}^\ast)\leq 0.483$ ($N=40, \delta=0.1$), and the numerical risk is lower than $0.483$ after $15$ iteration steps which is consistent with the anticipated bound.}

\begin{lemma}[\cite{mohri2018foundations}]
Fix a data distribution $(\bm x,y)\sim \mathcal{D}$, kernel function $Q$ and training data size $N$. Then, with probability at least $1-\delta$,
\begin{align}
R(h)\leq R_N(h)+\mathcal{O}\left(\sqrt{\frac{\log(2/\delta)}{N}}\right)
\end{align}
holds.
\label{lemma:gen}
\end{lemma}

\begin{proof}
Notice that if the quantum kernel $Q$ can be exactly calculated, then by Goel and Klivans~\cite{goel2019learning}, Quantum kernel Alphatron in the main file outputs a hypothesis $h^\ast$ such that
\begin{equation*}
  R(h^\ast)\leq \Ord{\sqrt{\varepsilon_g} + G\sqrt[4]{\frac{\log (1/\delta)}{N}} + B\sqrt{\frac{\log (1/\delta)}{N}}}.
\end{equation*}
This inequality is obtained by leveraging of
\begin{equation}
    R(h^\ast)\leq \hat{R}\pbra{h^{t\ast}} + \Ord{B\sqrt{\frac{1}{N}} + \sqrt{\frac{\log (1/\delta)}{N}}},
    \label{eq:err_ho}
\end{equation}
and
\begin{equation}
    \hat{R}\pbra{h^{t\ast}}\leq \Ord{\sqrt{\varepsilon_g} + G \sqrt[4]{\frac{\log (1/\delta)}{N}} + B\sqrt{\frac{\log (1/\delta)}{N}}}.
    \label{eq:errRh}
\end{equation}
for some $t^\ast\leq T = \Ord{N/\log(1/\delta)}$. Nevertheless, if the quantum kernel $Q$ is approximated via performing quantum circuits, Eq.~\eqref{eq:errRh} should be replaced with
\begin{equation}
        \hat{R}\pbra{\hat{h}^{t\ast}}\leq \Ord{\sqrt{\varepsilon_g} + G \sqrt[4]{\frac{\log (1/\delta)}{N}} + B\sqrt{\frac{\log (1/\delta)}{N}}}.
    \label{eq:errRhNew}
\end{equation}
where $\hat{h}^t\pbra{\bm x} = \sum_{i = 1}^m \alpha_i^t \hat{Q}\pbra{\bm{a_i}, \bm x}$, and $\hat{Q}$ is the approximation of $Q$. 

In the following, we will prove that $\abs{\hat{R}\pbra{\hat h^{t\ast}} - \hat{R}\pbra{h^{t\ast}}}$ is bounded, and hence $R\pbra{h^\ast}$
 is bounded by combining Eq.~\eqref{eq:err_ho},~\eqref{eq:errRh} and~\eqref{eq:errRhNew}.
 By Theorem~3 in the main file, $\hat{Q}\pbra{\bm{a_i}, \bm x}$ is an $\epsilon_Q$ approximation of $Q\pbra{\bm{a_i}, \bm x}$, \emph{i.e.},
\begin{equation}
    \abs{\hat{Q}\pbra{\bm{a_i}, \bm x} - Q\pbra{\bm{a_i}, \bm x}}\leq \epsilon_Q
\end{equation}
with high probability.

For convenience, in the later proof we require for all $i$, $\delta_Q^i = \hat{Q}\pbra{\bm{a_i}, \bm x} - Q\pbra{\bm{a_i}, \bm x}$ are the same, and $\delta_{\alpha_i}^t = \hat{\alpha}^t_i - \alpha^t_i$ are also the same, denoted them as $\delta_Q,\delta_{\alpha}^t$ respectively. Since $\delta_Q^i$ are in the same order for $i\in [N]$ ($\delta_{\alpha_i}^t$ similarly), hence it is reasonable for the assumptions. We will have the same upper bound for $R\pbra{h^\ast}$ without the assumptions and with a more tedious proof.
Then for any $i$, we have
\begin{align}
\begin{aligned}
       &-\delta_\alpha^t = \alpha_i^t - \hat{\alpha}_i^t \\
    &= \alpha_{i}^1 - \hat\alpha_{i}^1 +\frac{1}{N}\sum_{k = 1}^{t-1} \pbra{\hat{h}^{k}\pbra{\bm{a_i}} - h^{k}\pbra{\bm{a_i}}}\\
    &=\frac{1}{N}\sum_{k = 1}^{t-1} \sum_{j = 1}^N \pbra{\hat{\alpha}_j^k \hat Q\pbra{\bm{a_j}, \bm{a_i}} - \alpha_j^k Q\pbra{\bm{a_j}, \bm{a_i}}}\\
    &=  \sum_{k = 1}^{t-1} \pbra{A^k \delta_Q + \bar Q_i 
    \delta_\alpha^k + \delta_\alpha^k\delta_Q }
\end{aligned}
\label{eq:deltaRecur}
\end{align}
where $A^k = \frac{1}{N}\sum_{j = 1}^N \alpha_i^k$, and $\bar Q_i = \frac{1}{N}\sum_{j = 1}^N Q\pbra{\bm{a_j}, \bm{a_i}}$.
We can also obtain the value of $-\delta_{\alpha}^{t-1}$ by leveraging of Eq.~\eqref{eq:deltaRecur} and the recurrence relationship.  The following equations follows by subtracting $-\delta_{\alpha}^{t-1}$ by $-\delta_\alpha^t $,
\begin{align}
    \delta_\alpha^t = \pbra{\bar Q_i - 1}\delta_\alpha^{t-1} + A^{t-1}\delta_Q + \delta_\alpha^{t-1} \delta_Q,
\end{align}
hence with the fact that $0\leq \bar{Q}_i\leq 1$ and $A^k\leq \frac{k-1}{N}$, the absolute value of $\delta_\alpha^t$ satisfies the inequality
\begin{align*}
    \abs{\delta_\alpha^t}\leq\pbra{1 + \abs{\delta_Q}} \abs{\delta_\alpha^{t-1}} + \frac{t-2}{N} \abs{\delta_Q}\\
    =  \abs{\delta_\alpha^{t-1}} + \frac{3(t-2)}{N}\epsilon_Q
\end{align*}
By the recurrence of $\abs{\delta_\alpha^t}$, we have
\begin{align*}
    \abs{\delta_\alpha^t} &\leq \frac{3\epsilon_Q}{N}\sum_{k = 1}^{t - 2}k\\
    &\leq \frac{3t^2\epsilon_Q}{2N},
\end{align*}
then we have
\begin{align*}
    \abs{h^t\pbra{\bm x} - \hat h^t\pbra{\bm x}} &=\abs{\sum_{i = 1}^N  \pbra{\hat{\alpha}_i^t \hat{Q}\pbra{\bm{a_i},\bm x} - \alpha_i^t Q\pbra{\bm{a_i}, \bm x}}}\\
    &\leq \sum_{i = 1}^N \pbra{2\abs{\alpha_{i}^t} \epsilon_Q + 2 Q\pbra{\bm{a_i}, \bm x}|\delta_\alpha^t|}\\
    &\leq 2N\pbra{\frac{t-1}{N} \epsilon_Q + \abs{\delta_\alpha^t}}\\
    &\leq 4t^2\epsilon_Q,
\end{align*}
for large where the second inequality holds since $\abs{\alpha_i^t}\leq \frac{t-1}{N}$.

Therefore,
\begin{align*}
\abs{ \hat{R}\pbra{h^{t\ast}} -  \hat{R}\pbra{\hat{h}^{t\ast}} }&\leq \frac{1}{N}\abs{\sum_{i = 1}^{N}\pbra{h^t\pbra{\bm {a_i}} - \hat h^t\pbra{\bm {a_i}}} }\\
    &\leq 4t^2\epsilon_Q,
\end{align*}
where the firstly inequality holds by the definition of $\hat{R}\pbra{h^{t\ast}}$ and $\hat{R}\pbra{\hat{h}^{t\ast}}$ (Recall that $$\hat{R}\pbra{h^{t\ast}} = \frac{1}{N}\vabs{\sum_{i = 1}^N \pbra{\Ebb\sbra{b_i|\bm{a_i}} - h^t\pbra{\bm{a_i}}}\ket{\phi\pbra{\bm{a_i}}} }$$). 
The additive error $\epsilon_Q$ for the quantum kernel $Q\pbra{\bm{a_i}, \bm x}$ can be bounded to $\Ord{\frac{\sqrt{\log(1/\delta)}}{N^{5/4}}}$ with $\Ord{N^{5/2}}$ copies of the quantum states. 

Hence,
\begin{align*}
\begin{aligned}
    R(h^\ast)&\leq\Ord{\sqrt{\varepsilon_g} +  G \sqrt[4]{\frac{\log (1/\delta)}{N}} + B\sqrt{\frac{\log (1/\delta)}{N}} + {\frac{\log (1/\delta)}{N^{1/4}} }}\\
     &\leq \Ord{\sqrt{\varepsilon_g} +  G \sqrt[4]{\frac{\log (1/\delta)}{N}} + B\sqrt{\frac{\log (1/\delta)}{N}} },
\end{aligned}
\end{align*}
where the last inequality holds since $t =\Ord{N/\log(1/\delta)}$ and $G = \Omega\pbra{1}$. Combine the above inequality to Lemma~\ref{lemma:gen}, the upper bound of generalize error is obtained.
\end{proof}

\section{Implementation of quantum kernel with SWAP test}
\label{app:QkernelSWAP}
By leveraging of Chernoff bound, the quantum kernel can be approximated by independently performing the Destructive-Swap-Test~\cite{garcia2013swap} to $\mathcal{O}(\log (1/\delta)/\epsilon_Q^2)$ copies of $2n$-qubit state $\ket{\phi\pbra{\bm {a_i}}}\otimes \ket{\phi\pbra{\bm x}}$, with additive error $\epsilon_Q$ and failure probability $\delta$. The expectation of the measurement results of the Destructive-Swap-Test is
\begin{align}
   \langle\phi(\bm a_{\bm i})\otimes\phi(\bm x)|\mathbf{SWAP}|\phi(\bm a_{\bm i})\otimes\phi(\bm x)\rangle= Q(\bm a_{\bm i},\bm x),
\end{align}
where $\mathbf{SWAP}|\phi(\bm a_{\bm i})\otimes\phi(\bm x)\rangle=|\phi(\bm x)\otimes\phi(\bm a_{\bm i})\rangle$ denotes the $2n$-qubit swap operator.
For QPL problem, $\ket{\phi\pbra{\bm{a_i}}}$ and $\ket{\phi\pbra{\bm x}}$ can all be generated with polynomial-size circuit, hence the Destructive-Swap-Test can be performed efficiently.

\begin{figure*}[htb]
\centering
  \includegraphics[width=0.6\textwidth]{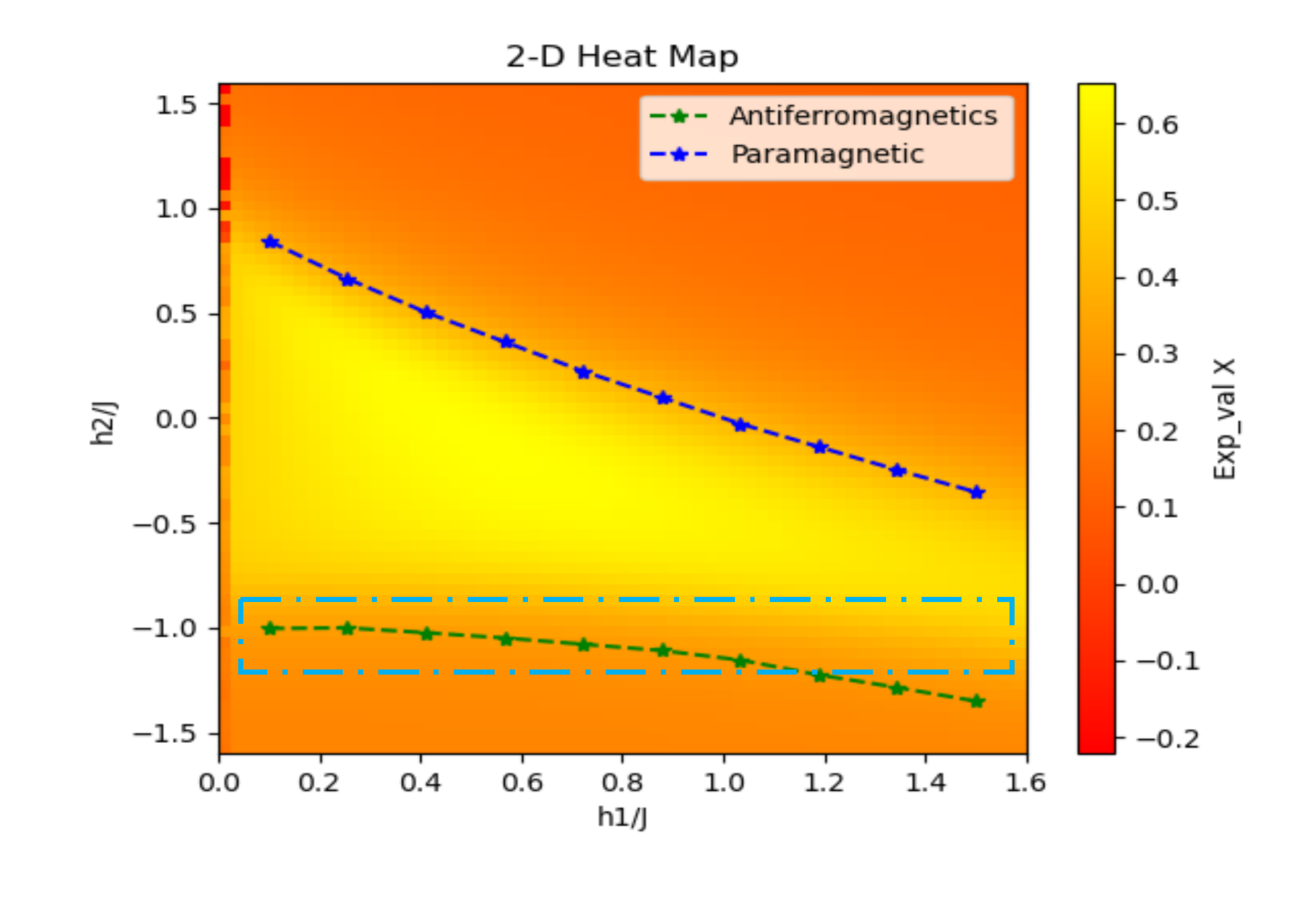}
  \caption{ Numerical results for QCNN to recognize a  $Z_2\times Z_2$ Symmetry-Protected-Topological (SPT) phase of Haldane Chain by using code~\cite{Jaybsoni2021}.}
  \label{Fig:QCNN}
\end{figure*}

\section{Discussions on sample complexity of ground states}
\label{App:SamplingComplexity}
Here, we carried out further theoretical analysis and numerical calculation on the distribution probabilities 
of ground states $|\psi_g(\bm x)\rangle$ of parameterized Hamiltonian $H(\bm x)$. We aim to provide a numerical window of $\bm x$, for which the corresponding Hamiltonian simulation is expected to be classically hard.

The probability distribution of a truly random quantum state $|\psi\rangle$ possesses the Porter-Thomas (PT) distribution ${\rm Pr}(|\langle\bm j|\psi\rangle|^2)=2^ne^{-2^n |\langle\bm j|\psi\rangle|^2}$, which is known to be classically hard to sample~\cite{brody1981random, boixo2018characterizing}. In the following we compare the distribution probabilities of ground states $|\psi_g(\bm x)\rangle$ of parameterized Hamiltonian $H(\bm x)$ with the Porter-Thomas distribution.

\begin{theorem}
Suppose $p_{\bm x}(\bm j)=|\langle\bm j|\psi_g(\bm x)\rangle|^2$ and $p(\bm j)=|\langle\bm j|\psi\rangle|^2$ represent probability distributions of ground state $|\psi_g(\bm x)\rangle$ and some random state $|\psi\rangle$, and their trace distance satisfies
\begin{align}
    {\rm Tr}\left({\rm Pr}(p_{\bm x}(\bm j)), {\rm Pr}(p(\bm j))\right)=\frac{1}{2}\sum\limits_{\bm j=0}^{2^n-1}|{\rm Pr}(p_{\bm x}(\bm j))-{\rm Pr}(p(\bm j))|<\epsilon.
    \label{KLD}
\end{align}
If the trace distance $\epsilon\leq n^{-1}$, then it is classically hard to sample from the ground state $|\psi_g(\bm x)\rangle$.
\label{thm:verify}
\end{theorem}


\emph{Proof Sketch:} Here, we provide a proof by contradiction. On one hand, it is known $\# P$-hard to approximate $p(\bm j)$ to additive error $\mathcal{O}(1/(n2^n))$ with a constant probability~\cite{bouland2019complexity}. On the other hand, if $\epsilon\leq n^{-1}$ and assume that there exists a classical sample algorithm $\mathcal{A}$ that can efficiently sample from $|\psi_g(\bm x)\rangle$, then $\mathcal{A}$ can be used to efficiently estimate $p(\bm j)$ to additive error $\mathcal{O}(1/(n2^n))$ with probability $\frac{1}{4}$. This  leads to a contradiction. Therefore, if $\epsilon\leq n^{-1}$, such classical sample algorithm $\mathcal{A}$ does not exist.

\begin{proof}
Assume that there exists a classical sample algorithm $\mathcal{A}$ that can efficiently sample from the ground state $|\psi_g(\bm x)\rangle$. Let $(\bm j_1,\bm j_2,...)$ be samples generated by the classical algorithm $\mathcal{A}$, then an approximation of ${\rm Pr}(p_{\bm x}(j))$ can be obtained by using Stockmeyer Counting theorem~\cite{boixo2018characterizing, stockmeyer1983complexity,goldreich2008computational}, which is denoted as ${\rm Pr}(\hat{p}_{\bm x}(\bm j))$. Then we have
\begin{equation}
\begin{aligned}
    |{\rm Pr}(\hat{p}_{\bm x}(\bm j))-{\rm Pr}(p(\bm j))|&\leq |{\rm Pr}(\hat{p}_{\bm x}(\bm j))-{\rm Pr}(p_{\bm x}(\bm j))|+|{\rm Pr}(p_{\bm x}(\bm j))-{\rm Pr}(p(\bm j))|\\
    &\leq \frac{{\rm Pr}(p_{\bm x}(\bm j))}{{\rm poly}(n)}+|{\rm Pr}(p_{\bm x}(\bm j))-{\rm Pr}(p(\bm j))|\\
    &\leq \left(1+\frac{1}{{\rm poly}(n)}\right)|{\rm Pr}(p_{\bm x}(\bm j))-{\rm Pr}(p(\bm j))|+\frac{{\rm Pr}(p(\bm j))}{{\rm poly}(n)}
\end{aligned}
\end{equation}
where the second inequality comes from the Stockmeyer Counting theorem and the third inequality comes from the triangle inequality. According to Markov's inequality, one has
\begin{align}
    {\rm Pr}\left(|{\rm Pr}(p_{\bm x}(\bm j))-{\rm Pr}(p(\bm j))|\leq\frac{2\epsilon}{2^n\delta}\right)\geq 1-\delta,
\end{align}
where $\delta\in[0,1]$.
Setting $\delta=\alpha\epsilon$, where $\alpha$ is a positive constant value, the relationship
\begin{align}
    |{\rm Pr}(\hat{p}_{\bm x}(\bm j))-{\rm Pr}(p(\bm j))|\leq \frac{1+1/{\rm poly}(n)}{\alpha2^{n-1}}+\frac{{\rm Pr}(p(\bm j))}{{\rm poly}(n)}
\end{align}
is valid with the probability $1-\alpha\epsilon$. If $\epsilon<\alpha^{-1}$ (which means the above estimation is valid), an estimation of ${\rm Pr}\left(p(\bm j)\right)$ is obtained, that is
\begin{align}
    \frac{{\rm Pr}(\hat{p}_{\bm x}(\bm j))-a}{1+b}\leq{\rm Pr}\left(p(\bm j)\right)\leq \frac{{\rm Pr}(\hat{p}_{\bm x}(\bm j))-a}{1+b}+\epsilon_1
\end{align}
where $a=(1+1/{\rm poly}(n))/(\alpha2^{n-1})$, $b=1/{\rm poly}(n)$ and $\epsilon_1=2a/(1-b^2)$. 
Since $p(\bm j)$ satisfies PT distribution, then we have
\begin{align}
    \frac{n-\log(l(\bm x)+\epsilon_1)}{2^n}\leq p(\bm j)\leq \frac{n-\log(l(\bm x)+\epsilon_1)}{2^n}+\frac{\epsilon_2}{2^n},
\end{align}
where $l(\bm x)=\frac{{\rm Pr}(\hat{p}_{\bm x}(\bm j))-a}{1+b}$ and $\epsilon_2=\log(1+\epsilon_1/l(\bm x))=\log(1+2a/(1+b)({\rm Pr}(\hat{p}_{\bm x}(\bm j))-a))$. For the random quantum state $|\psi\rangle$, the expectation value of ${\rm Pr}(p(\bm j))$ is upper bounded by $2^n\int_{0}^1pe^{-2^np}dp\leq2^{-n}$~\cite{boixo2018characterizing}. Then using Markov's inequality, we have  ${\rm Pr}(p(\bm j))\geq 1/2^n$ with probability $1/e$. Therefore ${\rm Pr}(\hat{p}_{\bm x}(\bm j))\geq {\rm Pr}(p(\bm j))-1/(\alpha 2^n)\geq 1/2^n-1/(\alpha 2^n)$ with probability $1/e$.

Taking everything together, the classical algorithm $\mathcal{A}$ combined with classical post-processing can provide an estimation to $p(\bm j)$ within an additive error
\begin{align}
    \frac{\epsilon_2}{2^n}\leq \frac{1}{2^n}\log\left(1+\frac{1}{\alpha2^{n-2}({\rm Pr}(\hat{p}_{\bm x}(\bm j))-a)}\right)\leq \frac{1}{2^n}\log\left(1+\frac{4}{\alpha}\right)\leq\frac{1}{\alpha 2^{n-2}}
\end{align}
with probability $(1-\alpha\epsilon)e^{-1}$. Let $\alpha=(1-e/4)\epsilon^{-1}$ ($(1-\alpha\epsilon)e^{-1}=\frac{1}{4}$), the above estimation is valid with probability $1/4$.
Therefore, under the conditions $\epsilon\leq n^{-1}$, and assuming $\mathcal{A}$ can efficiently sample from the ground state $|\psi_g(\bm x)\rangle$, then there exists a classical algorithm which can estimate $p(\bm j)$ to a $\mathcal{O}(1/(n2^n))$ additive error with a constant probability.

However, we already know that it is $\# P$-hard to estimate $p(\bm j)$ to a $\mathcal{O}(1/(n2^n))$ additive error with a constant probability. Then, if $\epsilon<n^{-1}$, the existence of such classical algorithm $\mathcal{A}$ would lead to Polynomial Hierarchy collapse to its second level~\cite{arora2009computational}. Therefore, no classical algorithm can efficiently sample from $|\psi_g(\bm x)\rangle$ if $\epsilon<n^{-1}$.
\end{proof}

To illustrate the usefulness of theorem~\ref{thm:verify}, we show numerical calculations for (1) the lattice Transverse-field Ising model and (2) the Fermi-Hubbard model. The lattice Transverse-field Ising model is given by
$H_I(W,J.F)=W\sum_{i=1}^nZ_i+J\sum_{(i,j)}Z_iZ_{j}-\frac{F}{2}\sum_{i=1}^nX_i$,
where $X_i$ and $Z_i$ are Pauli operators on the $i$-th qubit, and $\bm x=(W,J,F)$ determines the relative strength of the Hamiltonian terms. And the Fermi-Hubbard model is given by
$H_H(t,U)=-t\sum_{\langle i,j\rangle, s}(a^{\dagger}_{i,s}a_{j,s}+a^{\dagger}_{j,s}a_{i,s})+U\sum_{i}n_{i\uparrow}n_{i\downarrow}$, where $\bm x=(t,U)$ determines the relative strength of the Hamiltonian terms, $a_{i,s}^{\dagger}$ and $a_{i,s}$ are fermionic creation and annihilation operators, $n_{i\uparrow}=a^{\dagger}_{i\uparrow}a_{i\uparrow}$ and similarly for $n_{i\downarrow}$. The notation $\langle i,j\rangle$ in the first sum associates sites that are adjacent in a $(n_a\times n_b)$ lattice, and $s\in\{\uparrow, \downarrow\}$. Here, we denote the above two Hamiltonians as $H(\bm x)$, and analyze the probability distribution of $|\psi_g(\bm x)\rangle$, namely the ground state of $H(\bm x)$.

\begin{figure*}[htb]
\centering
  \includegraphics[width=\textwidth]{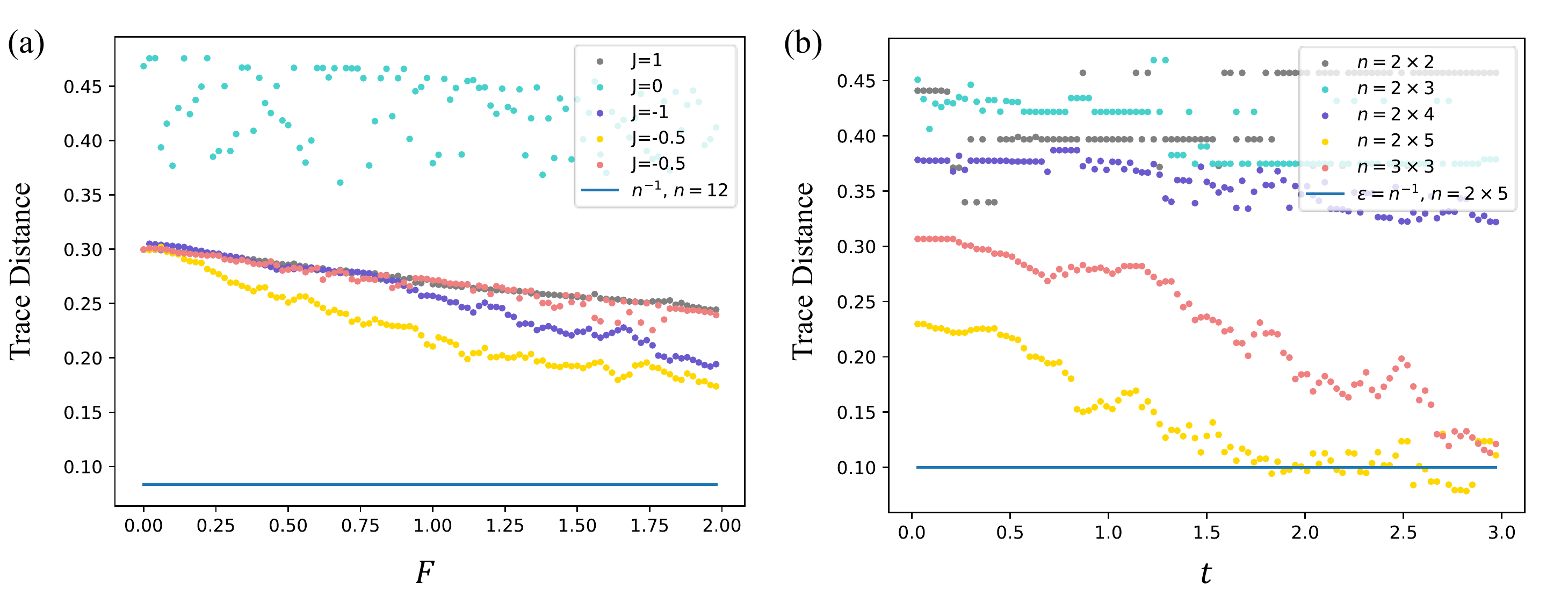}
  \caption{Numerical results to illustrate the sample complexity of $|\psi_g(\bm x)\rangle$. Each dot represents the KLD between the
  probability distribution ${\rm Pr}(p_x(\bm j))$ and PT distribution for (a) 2-dimensional lattice Hamiltonian and (b) 2-dimensional Fermionic Hubbard Hamiltonian.}
  \label{Fig:Amplitude}
\end{figure*}

We use the trace distance ${\rm Tr}(p, q)$ as a measure between two distributions, where ${\rm Tr}(p,q)\in[0,1]$; ${\rm Tr}(p,q)=0$ holds if and only if distribution $p=q$. In Figure~\ref{Fig:Amplitude}(a), we plot our results for lattice Transverse-field Ising model with $n=12$, $W=1$, and a range of $J$ and $F$ values. When the nearest neighbour coupling strength $J=0$, the corresponding ground states are classically solvable. sample from these ground states is easy for classical algorithms, as indicated by large values of trace distance in Figure~\ref{Fig:Amplitude}(a). On the other hand, $J\neq 0$ induces complex ground states $|\psi(\bm x)\rangle$, resulting much smaller values of  trace distances. However, the ground states of 2-dimensional lattice models donot saturate into the $\epsilon\leq n^{-1}$ domain which has been proved to be classically hard. 

In Figure~\ref{Fig:Amplitude}(b), we plot ground states of Hubbard model for $U=1$, and a range of $n=(n_a\times n_b)$ and $t$ values. We observe that the change of $t$ increases the sample complexity of $|\psi_g(\bm x)\rangle$ for different size of the Hubbard models. The ground states of $(2\times 5)$ Fermi-Hubbard model saturate into the $\epsilon=n^{-1}=0.1$ error. It also demonstrates Fermi-Hubbard ground states are more difficult compared with (a) lattice Ising model ground states, which is consistent with physical intuition. 

\vspace{10px}

\section{Numerical comparison to related works}
\subsection{Comparison to QCNN}
To provide a fairly comparison, we first elaborately provide the computational overhead of both method in each iteration step. For the proposed Algorithm 1, the quantum learner should first calculate an $N\times N$ quantum kernel method with $\mathcal{O}\left(N^2/\epsilon^2\right)$ sample complexity and a quantum circuit with $\mathcal{O}(n)$ controlled Z gates. After that, the proposed Algorithm 1 does not need any quantum resource in each iteration step. The QCNN loss function is
$$L(U,V,F)=\frac{1}{N}\sum\limits_{i=1}^N\left(y_i-f_{U,V,F}(|\psi(\bm x_i)\rangle)\right)^2,$$
where $f_{U,V,F}$ denotes the output of QCNN and $U, V, F$ represent the variational quantum circuits in QCNN. At depth $d$, the QCNN method requires $\mathcal{O}(\frac{7n}{2}(1-3^{1-d})+n3^{1-d})$ multi-qubit operations and $4d$ single-qubit rotations to provide an output. Repeat this quantum circuit $N$ times, the QCNN obtains a loss function $L(U,V,F)$ in a single iteration step. 

Consider a classification task on a set $C=\{c_1,c_2\}$ of $2$ classes in a supervised learning scenario. In such settings, a training set $\mathcal{S}$ and a testing set $\mathcal{T}$ both are assumed to be labeled by a map $m:\mathcal{S}\cup\mathcal{T}\mapsto C$, and both $\mathcal{S}$ and $\mathcal{T}$ are provided to the learner, where only the training set $\mathcal{S}$ has the label. Formally, the learner has only access to a restriction $\widetilde{m}(\bm x)$ of the indexing map:
\begin{align}
   \widetilde{m}: \mathcal{S}\mapsto C.
\end{align}
Suppose the learner outputs a model $h(\cdot)$ to predict the label of data $\bm x\in \mathcal{T}$, and the classification result can be defined as:
\begin{equation}
\label{eq6}
m(\bm x)=\left\{
\begin{aligned}
c_1 & , & h(\bm x)>t_1, \\
c_2 & , & h(\bm x)<t_2,
\end{aligned}
\right.
\end{equation}
where $t_1$ and $t_2$ are selected thresholds. Then the accuracy of the model is quantified by a classification success rate, proportional to the number of collisions:
\begin{align}
    v_{s}=\frac{|\{x\in\mathcal{T}|m(\bm x)=\widetilde{m}(\bm x)\}|}{|\mathcal{T}|}.
\end{align}
In the classification of SPT cases, the training data set $\mathcal{S}$ contains $40$ points on the line of $h_2=0$, and the testing data set $\mathcal{T}$ contains $4096$ points uniformly distributed on the domain $(h_1/J,h_2/J)\in[-1.5,1.5]\times[0,1.6]$. 

We test the proposed Algorithm 1 on $\mathcal{T}=\{\bm x_i,y_i\}$ where $61$ points are mis-classified in the vicinity of paramagnetic boundary, and the classification accuracy $v_s=0.985$ in this case. Since the Ref~\cite{cong2019quantum} does not provide the exactly accuracy of QCNN, we thus simulate the QCNN based on the code in the link (https://github.com/Jaybsoni/Quantum-Convolutional-Neural-Networks). The simulation results show that its output cannot perfectly fit the Antiferromagnetic boundary which is consistent with Fig. 4 in ref~\cite{cong2019quantum}. Using the same testing set with $4096$ data points, there are $116$ points are mis-classified, and its classification accuracy $v_s=0.971$. The predicted phase diagrams is illustrated as Fig~\ref{Fig:QCNN}.

We also numerically show our method takes less steps to a relatively stable landscape compared to the QCNN method (see Fig.~\ref{Fig:loss}).

\begin{figure*}[htb]
\centering
  \includegraphics[width=0.6\textwidth]{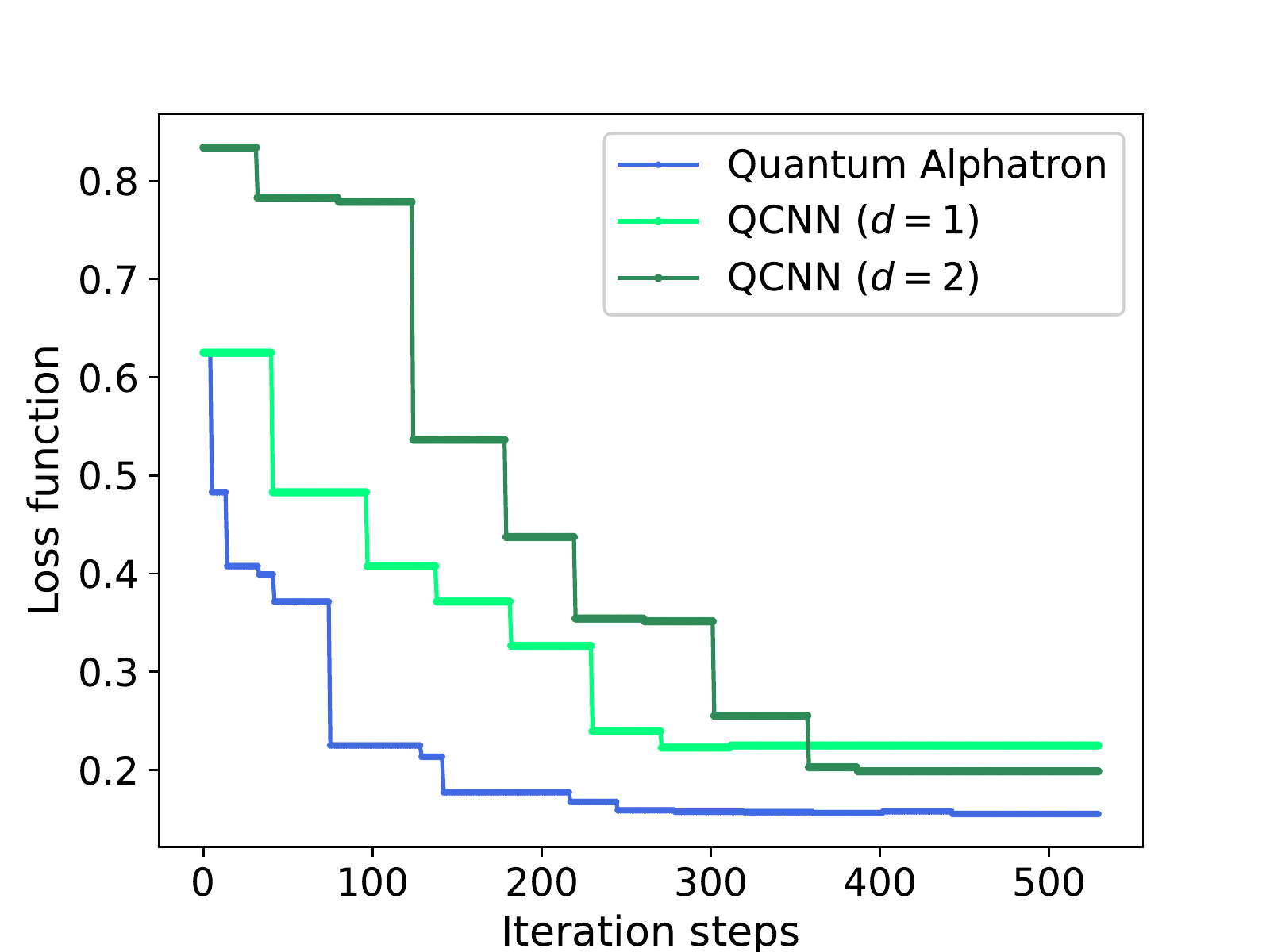}
  \caption{The three curves indicate the variation trend of the loss function $\hat{R}_L(h^\ast)$ in the training procedure.}
  \label{Fig:loss}
\end{figure*}

\subsection{Comparison to Ref.~~\cite{huang2021provably}}
Here, we utilize shadow-tomography based kernel-PCA method~\cite{huang2021provably} and common kernel-PCA method~\cite{hoffmann2007kernel} in recognizing $Z_2\times Z_2$ SPT phase of Haldane chain and three distinct phases of bond-alternating XXZ model.

We first consider the method proposed in~\cite{huang2021provably} that designed a special kernel function for classical shadows. Given an unknown density matrix $\rho$, implementing the randomrized single-qubit Pauli measurements, a classical shadow representation is obtained, where each shadow raw data corresponds to a two-dimensional array
\begin{align}
S_T(\rho)=\left\{|s_{i}^{t}\rangle: i \in\{1,...,n\}, t\in\{1,...,T\}\right\}\in \{|0\rangle,|1\rangle,|+\rangle,|-\rangle,|i+\rangle,|i-\rangle\}^{n\times T}.
\end{align}
And a classical representation of the density matrix $\rho$ is 
\begin{align}
\sigma_T(\rho)=\frac{1}{T}\sum\limits_{t=1}^T\sigma_1^{(t)}\otimes\sigma_2^{(t)}\otimes...\otimes \sigma_n^{(t)},
\end{align}
where $\sigma_i^{(t)}=3|s_i^{t}\rangle\langle s_i^t|-I$. Then the kernel function for shadow tomography can be computed by
\begin{align}
k^{\rm shadow}(\rho_1,\rho_2)=\exp\left(\frac{\tau}{T^2}\sum\limits_{t_1,t_2=1}^T\exp\left(\frac{\gamma}{n}\sum\limits_{i=1}^n{\rm Tr}\left(\sigma_i^{t_1}\sigma_i^{t_2}\right)\right)\right),
\end{align}
which is determined by hyper-parameters $\tau$ and $\gamma$.
Then perform kernel-PCA method on the kernel matrix $k^{\rm shadow}$, the classification result of quantum phases is obtained. Another strategy is directly applying kernel-PCA method on the provided shadow tomography, where the corresponding kernel matrix 
\begin{align}
k(\rho_1,\rho_2)=\exp\left(\tau{\rm Tr}\left(\rho_1\rho_2\right)\right)
\end{align}

Here, we test their performances in classifying $Z_2\times Z_2$ symmetry-protected-topological (SPT) phase of Haldane chain and three distinct phases of bond-alternating XXZ model. Let the sample complexity $T=500$, and simulation results are shown in Fig.~\ref{Fig:ML_for_spt} and Fig.~\ref{Fig:ML_for_xxz}. Noting that the shadow tomography kernel-PCA can recognize different quantum phases of matter, including SPT phase, symmetry-broken phase and trivial phase, while it is hard for common kernel-PCA method. Meanwhile, it is observed that shadow tomography kernel-PCA cannot provide a comparable classification result with QKA method (see Fig.~\ref{Fig:ML_for_spt} (c) and Fig.~\ref{Fig:ML_for_xxz} (c)), where some strange quantum phase assignments appear. For example, Fig.~\ref{Fig:ML_for_spt} (a) appears two mis-classified ground states in the top right corner which do not represent any quantum phase, and a similar phenomenon is observed on the left side in Fig.~\ref{Fig:ML_for_xxz} (a).

\begin{figure*}[htb]
\centering
  \includegraphics[width=0.8\textwidth]{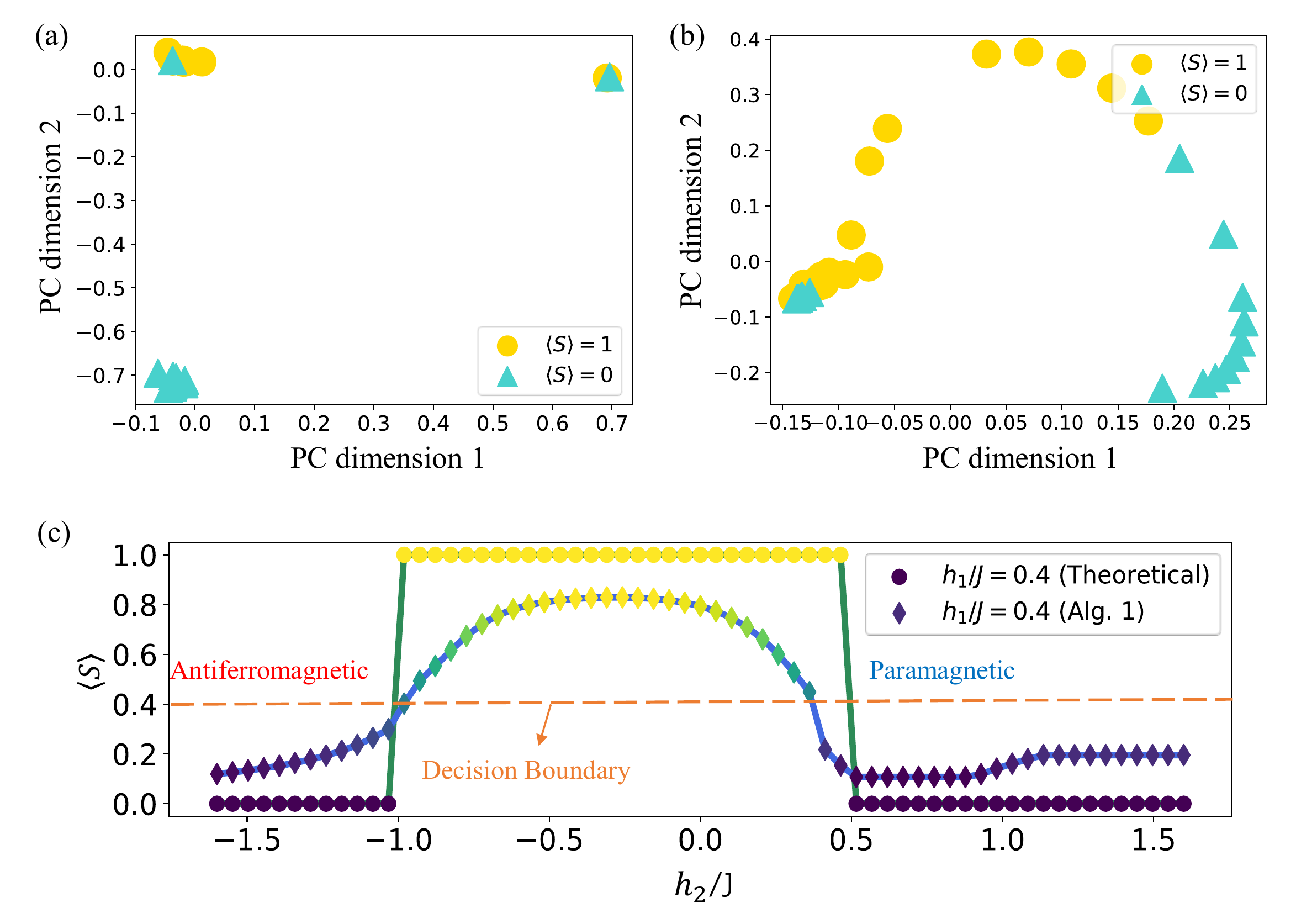}
  \caption{Numerical results for recognizing $Z_2\times Z_2$ Symmetry-Protected-Topological (SPT) phases of Haldane Chain at the cross-section $h_1/J=0.4$. (a) Classification results by using shadow tomography-based kernel-PCA ~\cite{huang2021provably} method by using code in (https://github.com/hsinyuan-huang/provable-ml-quantum). (b) Classification results by using kernel-PCA. (c) Prediction results by using quantum kernel method.}
  \label{Fig:ML_for_spt}
\end{figure*}

\begin{figure*}[htb]
\centering
  \includegraphics[width=0.8\textwidth]{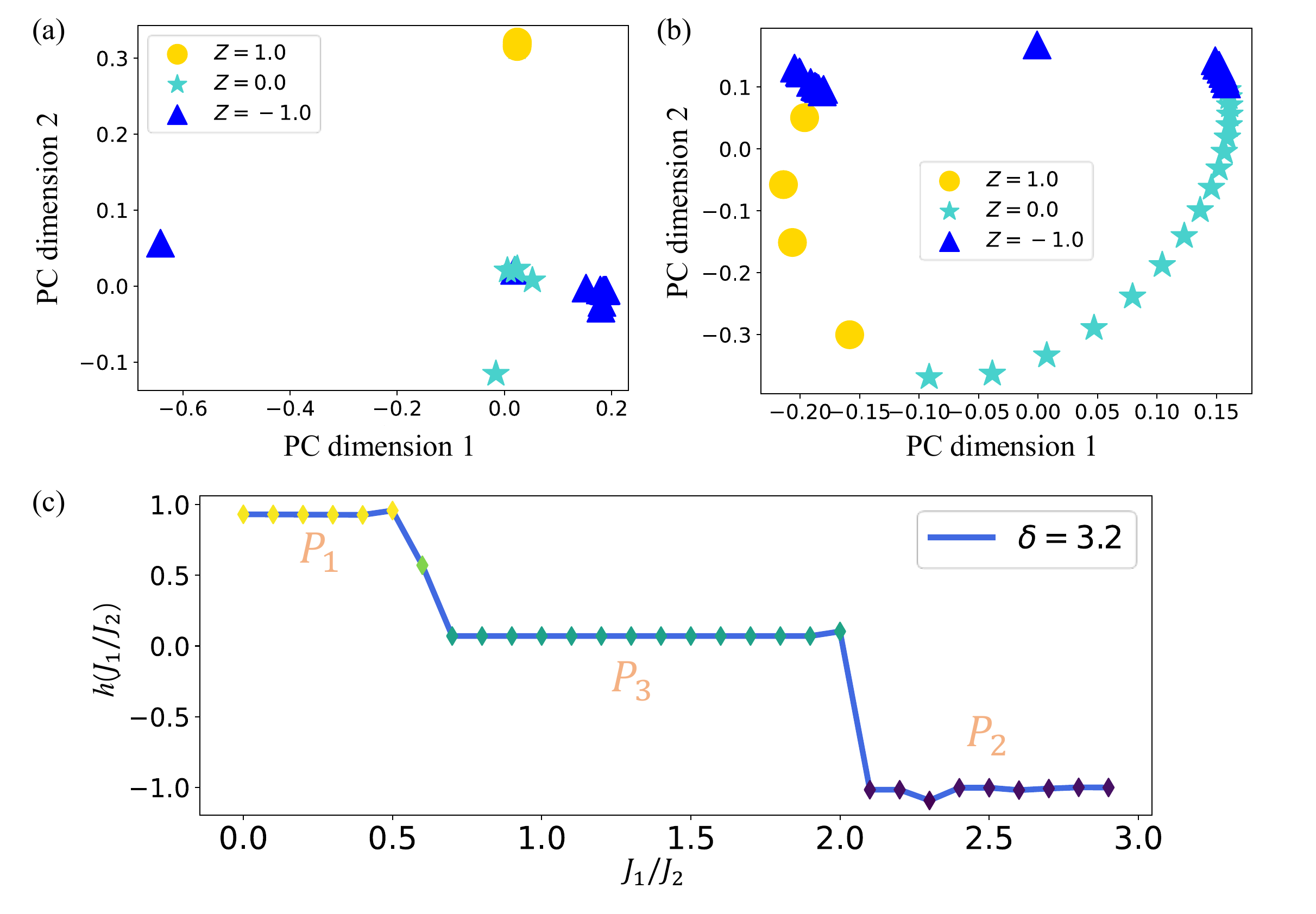}
  \caption{ Numerical results for recognizing three distinct phases of bond-alternating XXZ model at the function $\delta=3.2$. The invariant $Z=1$ marks the trivial phase, $Z=0.5$ marks symmetry broken phase and $Z=0$ marks the topological phase. (a) Classification results by using shadow tomography kernel PCA method. (b) Classification results by using kernel-PCA. (c) Prediction results by using quantum kernel method.
 }
  \label{Fig:ML_for_xxz}
\end{figure*}

\section{Review of Alphatron algorithm}
\label{app:alphatron}

In this section, we review the Alphatron algorithm~\cite{goel2019learning}, and give the comparison for Alphatron, Quantum Alphatron~\cite{rebentrost2021quantum}, and the Quantum kernel Alphatron algorithm (this paper).
\begin{algorithm}
\SetKwInOut{Input}{Input}
\SetKwInOut{Output}{Output}
\SetKwFor{While}{while}{do}{}%
\SetKwFor{For}{for}{do}{}
\Input{training set $\mathcal{S}=\{(\bm a_{\bm i},b_i)\}_{i=1}^N\in\mathcal{R}^d\times [0,1]$, non-decreasing $L$-Lipschitz function $u: \Rcal \rightarrow [0,1]$, kernel function $\Kcal$, learning rate $\lambda>0$, number of iterations $T$, testing data $\mathcal{S}=\{(\bm x_{\bm i},y_i)\}_{i=1}^M\in\mathcal{R}^d\times [0,1]$}
$\bm \alpha^{\bm 1}:=0\in\mathcal{R}^N$\;
 \For{$t=1,2,...,T$}{
$\hat{h}^{t}(\bm x):=\sum_{i=1}^N\alpha_i^t \Kcal(\bm a_{\bm i},\bm x)$\;
\For{$i=1,2,...,N$ }{
$\alpha_i^{t+1}=\alpha_i^t+\frac{\lambda}{N}(b_i-h^t(\bm a_{\bm i}))$\;
}
}
\emph{Let $r=\arg \min_{t\in\{1,...,T\}}\sum_{j=1}^M\pbra{h^t(\bm x_j)-y_j}^2$}\;
\Return{$h^r$}
\caption{Alphatron algorithm}
\label{alg:ori_alphatron}
\end{algorithm}

\begin{theorem}[\cite{goel2019learning}]
Let $\Kcal$ be a kernel function correpsonding to feature map $\psi$ such that $\forall \bm x \vabs{\psi(\bm x)}\leq 1$. Consider samples $\pbra{\bm a_i, b_i}_{i =1 }^N$ drawn iid from distribution $\Dcal$ on $\Xcal \times [0,1]$ such that $E[y|\bm x]=u\pbra{\langle\bm v, \psi(\bm x)\rangle} + \xi(\bm x)$ where $u:\Rcal \rightarrow [0,1]$ is a known $L$-Lipschitz non-decreasing function, $\xi:\Rcal^{d}\rightarrow [-G, G]$ for $M>0$ such that $\Ebb\sbra{\xi(\bm x)^2}\leq \epsilon$ and $\vabs{\bm v}\leq B$. Then for $\delta\in(0,1)$, with probability $1-\delta$, Alphatron with $\lambda = 1/L, T = CBL\sqrt{N/log(1/\delta)}$ and $M=C'N\log(T/\delta)$ for large enough constants $C,C'>0$ outputs a hypothesis $h$ such that, 
\begin{align}
    \varepsilon(h)\leq \Ord{L\sqrt{\epsilon} + LG \sqrt[4]{log(1/\delta)/N} + BL \sqrt{\log(1/\delta)/N}}.
    \label{eq:alphatron_risk}
\end{align}

\end{theorem} 
Alphatron algorithm requires poly$\pbra{N,d,\log(1/\delta), t_\Kcal}$ running time to train a kernel model, where $t_\Kcal$ is the running time on computing the kernel function. 
In Quantum Kernel Alphatron algorithm, we let the kernel be quantum kernel $Q(\bm x, \bm{a_i}) = \abs{\bra{0^n}U(\bm x)^\dagger U\pbra{\bm{a_i}}\ket{0^n}}^2$. Since $Q\pbra{\bm x, \bm{a_i}}$ is approximated via SWAP test, the risk in Eq. \eqref{eq:alphatron_risk} does not hold. Nevertheless, we prove that via $\Ord{N^{5/2}}$ copies of quantum states for each training data, the risk of the Quantum kernel Alphatron is also bounded. 

As a comparison, Quantum Alphatron algorithm~\cite{rebentrost2021quantum} quantizes the Alphatron algorithm to provide a quantum implementation of the well-known polynomial kernel function, which accelerate the running time of the original Alphatron algorithm in the scenario that the dimension $d$ of the data is much larger than other parameters.

\clearpage

\end{document}